\documentclass[11pt]{amsart} 
\usepackage{amsmath,amsfonts,amssymb,geometry}
\usepackage{graphicx,amsmath, amscd}
\usepackage{color}

\newtheorem{theorem}{Theorem}[section]
\newtheorem{lemma}[theorem]{Lemma}
\newtheorem{corollary}[theorem]{Corollary}
\newtheorem{proposition}[theorem]{Proposition}
\newtheorem{definition}[theorem]{Definition}
\newtheorem{remark}[theorem]{Remark}

\newtheorem{claim}[theorem]{Claim}
\newtheorem{conjecture}[theorem]{Conjecture}
\newcommand{\Id}{\mathbb{I}}

\newcommand{\K}{\cS_{\textnormal{sym}}}

\renewcommand{\leq}{\leqslant}
\renewcommand{\geq}{\geqslant}

\DeclareMathOperator{\vrad}{\mathrm{vrad}}
\DeclareMathOperator{\vol}{\mathrm{vol}}
\DeclareMathOperator{\conv}{\mathrm{conv}}
\DeclareMathOperator{\Span}{\mathrm{span}}
\DeclareMathOperator{\tr}{\mathrm{tr}}
\DeclareMathOperator{\diag}{\mathrm{Diag}}
\DeclareMathOperator{\spec}{\mathrm{sp}}

\newcommand{\N}{\mathbf{N}}

\newcommand{\R}{\mathbf{R}}
\newcommand{\C}{\mathbf{C}}
\newcommand{\Hone}{\cM_{n}^{sa,1}}
\newcommand{\Hzero}{\cM_{n}^{sa,0}}

\DeclareMathOperator{\E}{\mathbf{E}}
\renewcommand{\P}{\mathbf{P}}
\newcommand{\mL}{\mathcal{L}}
\newcommand{\mU}{\mathcal{U}}
\newcommand{\mF}{\mathcal{F}}
\newcommand{\SEP}{\cS_0}
\newcommand{\st}{\  : \ }

\newcommand{\be}{\begin{equation}}
\newcommand{\ee}{\end{equation}}
\newcommand{\bl}{\begin{lemma}}
\newcommand{\el}{\end{lemma}}
\newcommand{\eprop}{\end{proposition}}
\newcommand{\bprop}{\begin{proposition}}
\newcommand{\e}{\varepsilon}
\newcommand{\s}{\sigma}
\newcommand{\cD}{\mathcal{D}}
\newcommand{\cH}{\mathcal{H}}
\newcommand{\cK}{\mathcal{K}}
\newcommand{\cM}{\mathcal{M}}
\newcommand{\cS}{\mathcal{S}}
\newcommand{\cU}{\mathcal{U}}
\newcommand{\bC}{\mathbf{C}}
\newcommand{\iy}{\infty}

\begin{document}

\title{Entanglement thresholds for random induced states}

\author{Guillaume Aubrun}
%
\author{Stanis{\L}aw J. Szarek}
%
\author{Deping Ye}

\begin{abstract}

For a random quantum state on $\cH=\C^d \otimes \C^d$ obtained by
partial tracing a random pure state on $\cH \otimes \C^s$, we
consider the question whether it is typically separable or
typically entangled. For this problem, we show the existence of a
sharp threshold $s_0=s_0(d)$
 of order roughly $d^3$.
More precisely, for any $\e>0$ and for $d$ large enough, such a
random state is entangled with very large probability when $s \leq
(1-\e)s_0$,
 and separable with very large probability  when $s \geq (1+\e) s_0$.
One consequence of this result is as follows:  for a system of $N$
identical particles in a random pure state, there is a threshold
$k_0 = k_0(N)\sim N/5$ such that
 two subsystems of $k$ particles each typically share entanglement if  
$k>k_0$,  and  typically do not share entanglement if $k<k_0$. Our
methods work  also for multipartite systems and for ``unbalanced''
systems such as $\C^{d_1} \otimes \C^{d_2} $, ${d_1} \neq {d_2} $.
The arguments rely on random matrices, classical convexity,
high-dimensional probability and geometry of Banach spaces; some
of the auxiliary results may be of reference value.

\end{abstract}

\maketitle

\section{Introduction}

In recent years, random constructions have become a very fruitful
tool in quantum information theory. The study of random channels
and random states has particularly intensified since the
influential paper \cite{Hayden2006}. The most successful
achievement of the probabilistic method in quantum information
theory is arguably Hastings's proof that suitably chosen random
channels provide a counterexample to the additivity conjecture for
classical capacity of quantum channels \cite{Hasting2009}.

In this paper, we address the most fundamental question one can
ask about a random state: {\em is it entangled?} Detecting and
exploiting entanglement, first discovered in the 1930's
\cite{EPR1},   is a central problem in quantum information and
quantum computation at least since Shor's work \cite{Sh1} on
integer factoring. However, the structure of the set of entangled
quantum states is still not well-understood. The well-known
Peres--Horodecki positive partial transpose (PPT) criterion
\cite{Ho1, Pe1} is a necessary condition for (lack of)
entanglement, but this condition is sufficient only for
qubit-qubit and qubit-qutrit systems \cite{St1, Wo1}.

We consider here a family of random states that are known as
random induced states. These are mixed states on $\cH$ obtained
after partial tracing, over some ancilla space $\cH_a$, a
uniformly distributed pure state on $\cH \otimes \cH_a$.
This leads to a 
natural family of probability measures on the set of states on
$\cH$ (see \cite{Zycz2,Ben1}), where $\dim \cH_a$, the dimension
of the environment,  is a parameter. Indeed, if all that we know
about the system $\cH \otimes \cH_a$ are the dimensions of the
factors $\cH$ and $\cH_a$, and that it is isolated from the rest
of the environment, the corresponding random induced state is a
reasonable model for, or at least a reasonable first guess about
the state of the system $\cH$.

Of course, the induced state $\rho$ being random, we cannot expect
to be able to tell what $\rho$ is. However, we may be able to
infer some properties of $\rho$ if they are generic (that is,
occur with probability close to $1$) for a given random model. For
specificity, consider $\cH=\C^d \otimes \C^d$ and let us focus on
the question {\em ``Is a random state entangled?''} As it turns
out, the answer depends in a crucial and rather precise way on the
size of the environment with low-dimensional environments leading
typically to entangled states and high-dimensional environments
leading to separable states.

In the special case $\dim \cH_a=\dim \cH =d^2$, we are led to the
uniform distribution on the set of states (i.e.,  uniform with
respect to the usual Hilbert--Schmidt volume). As was shown in
\cite{AS1},  the proportion of states (again, measured with respect to
the Hilbert--Schmidt volume) that are separable is extremely small
in large dimensions. This means that when $\dim \cH_a = d^2$,
random induced states are typically entangled. This was extended
to the case when $\dim \cH_a$ is slightly larger than $d^2$ in
\cite{Deping2010,Aubrun2011}.

On the other hand, it was proved in \cite{Hayden2006} that random
induced states on $\C^d \otimes \C^d$ are typically separable when
$\dim \cH_a$ is proportional to $d^4$. In the present paper, we
bridge the gap between these estimates and show that the threshold
$s_0(d)$ between separability and entanglement occurs at order
$d^3$ (more precisely, we obtain the inequalities $cd^3 \leq
s_0(d) \leq Cd^3 \log^2 d$, where $C,c>0$ are universal constants,
independent of the dimensions involved). More specifically, we
show that for any $\e>0$, the following holds if  $d$ is large
enough. When the environment dimension $s=\dim \cH_a$ is smaller
than $(1-\e)s_0$, the
 random induced state is entangled with overwhelming probability.
When the environment dimension exceeds $(1+\e)s_0$, the  random
induced state is separable with overwhelming probability.

The heuristics behind the consequence stated in the abstract are 
now as follows. If we have a system of $N$ particles (with $D$
levels each) which is in a random pure state, and two subsystems
of $k$ particles each, then the ``joint state''  of the subsystems
is modeled by a random induced  state on $\C^d \otimes \C^d$ with
$d=D^k$  and $s = D^{N-2k}$. In particular, the relation $k=N/5$,
or $N=5k$, corresponds exactly to $s=d^3$. The reason for the
threshold effect is that passing from $k$ to $k-1$ increases $s$
by a factor of $D^2\geq 4$, while -- as stated above --  the
transition from generic  entanglement to generic  separability
takes place when $s$ is increased only slightly (the simultaneous
decrease of   $d$ by a factor of $D\geq 2$ only amplifies the
effect).

Our method of proof is geometric and uses tools from
high-dimensional convexity, a field also known as asymptotic
geometric analysis.  This has become a fruitful approach to study
the geometry of quantum states in large dimension; recent
contributions include, for instance, \cite{AubrunSzarekWerner2011,
SzarekWerner2011,fhs}. Asymptotic geometric analysis strives to
understand properties of geometric structures in high dimension,
which -- because of the ``central limit theorem-like'' effects --
are believed to be relatively easier to  pinpoint than those in
low dimension. The phenomenon of concentration of measure plays a
central role in the whole theory.

While our starting point is an estimation of the Hilbert--Schmidt
volume of the set $\cS$ of separable states on $\C^d \otimes
\C^d$, the relevant geometric parameter is not the
volume itself, but rather the mean width of the dual set
$\cS^\circ$. We estimate it in an indirect way, using the
$MM^*$-estimate, a general result from asymptotic geometric
analysis.

Our method also applies to the case of multipartite systems
$(\C^d)^{\otimes k}$. In the asymptotic limit,  when $k$ is fixed
and $d$ goes to infinity, we show that the threshold for
separability vs. entanglement occurs for $s$ of order $d^{2k-1}$
(up to logarithmic factors). However, to keep the exposition as
simple as possible, we focus on the bipartite case $k=2$ in most
of the paper.

Finally, we also show that ``one half'' of our main result
(entanglement is generic for small $\dim \cH_a$) can be proved in
a more elementary way by working with the densities of induced
measures. Similar arguments appeared already in the papers
\cite{Deping2009,Deping2010}.

 This paper is organized as follows. In Section
\ref{section:notation}, we introduce some mathematical background,
necessary notation, state our main theorem and present a high-level 
overview of the proof.  A rigorous proof is found in Sections
 \ref{section:threshold} and \ref{section:estimate-s0}. 
 To not to obscure the structure of the proof, some general (but 
technically involved) mathematical tools are collected and/or developed in
several Appendices, some of which may be of independent interest
and/or reference value. 
In Section \ref{section:multipartite},
we extend our results to the case of multipartite systems. In
Section \ref{section:alternative-argument}, we will provide
another, more elementary, proof of our result, which is valid only
in the regime where entanglement is generic. Miscellaneous remarks
and loose ends are addressed in Section \ref{misc}. 

General references for concepts related to
quantum information theory are \cite{Ben1, Nie1}, for
those related to asymptotic geometric analysis \cite{MS, Pisier1989, ledoux},
and for those related to random matrices \cite{agz, ds-handbook}.  
A high-level non-technical overview of the
results of this paper and of a related article \cite{Aubrun2011}
can be found in \cite{AubrunSzarekYeShort}.

\section{Notation, background and the statement of the main theorem}
 \label{section:notation}

The letters $C,c,c_0,...$ denote absolute numerical constants
(independent of the instance of the problem) whose values may
change from place to place. When $A,B$ are quantities depending on
the dimension (and perhaps some other parameters), the notation $A
\lesssim B$ means that there exists an absolute constant $C>0$
such that the inequality $A \leq CB$ holds in every dimension.
Similarly $A \simeq B$ means both $A \lesssim B$ and $B \lesssim
A$. As usual, $A\sim B$ means that  $A/B\to 1$ as the dimension
(or some other relevant parameter) tends to $\infty$, while $A =
o(B)$ means that $A/B\to 0$.

There are various dimensions appearing repeatedly in this paper,
and we will stick to the following notational scheme. We will work
in a complex Hilbert space $\cH=\C^n$, with $n=d^k$, so that we
may identify $\C^n$ with $(\C^d)^{\otimes k}$ ($k=2$ in most of
the paper). We will always assume that $d\geq 2$. The set of
states on $\C^n$ has real dimension $m=n^2-1$. We also consider an
ancilla space $\C^s$, and pure states in $\C^n \otimes \C^s$.
Occasionally the unit sphere in $\C^n \otimes \C^s$ will be
identified with the unit sphere $S^{m-1} \subset \R^m$, with
$m=2ns$. General results from convex geometry will also be stated
in $\R^m$ (except in Appendices \ref{app:wasserstein} and
\ref{app:marcenko-pastur-vs-GUE}, where we consider  $\R^{\dim
\cH}$,  and so $\R^n$ is more appropriate).

\subsection{Basic facts from convex geometry}\label{geometry}

A {\em convex body}  $K$ in a real finite-dimensional vector space
(usually identified with $\R^m$) is a convex compact set with
non-empty interior. We denote by $|\cdot|$ the Euclidean norm. A
convex body $K$ is {\em symmetric} if $K=-K$. The convex bodies we
consider may be non-symmetric, but we usually ``arrange'' that the
origin belongs to the interior of $K$. The {\em gauge}  associated to
such $K$ is the function $\|\cdot\|_K$ defined for $x \in \R^m$ by
\[ \|x\|_K := \inf \{ t \geq 0 \st x \in tK \} .\]
If $K$  is symmetric, $\|\cdot\|_K$ is a norm 
and $K$ is precisely the unit ball in that norm. 
However,  in general, we may have $\|x\|_K \neq \|-x\|_K$ when $K$ is
non-symmetric. 

By $\vol$ we denote the Lebesgue measure on $\R^m$. The {\em volume
radius} of a convex body $K \subset \R^m$ is defined as
\[ \vrad (K) := \left( \frac{\vol K}{\vol B_2^m} \right)^{1/m} ,\]
where $B_2^m$ denotes the unit Euclidean ball. In words,
$\vrad(K)$ is the radius of the Euclidean ball with same volume as
$K$. The same notation will be used if $K$ ``lives'' in an
$m$-dimensional linear or affine subspace of a larger space.

If $K \subset \R^m$ is a convex body with origin in the  interior,
the {\em polar} of $K$ is the convex body $K^\circ$ defined as
\[ K^\circ := \{ y \in \R^m \st \langle x,y \rangle \leq 1 \ \ \hbox{\rm for all  } \ x \in K \} .\] 
A basic result from convex analysis is that $(K^\circ)^\circ = K$ 
(this is the {\em bipolar theorem}, a baby version of the Hahn--Banach theorem). 

The {\em inradius} of a convex body $K$ is the largest radius $r$ of a
Euclidean ball contained in $ K$.
 Similarly, the {\em outradius} of $K$ is the smallest $R>0$ such that $K$ is contained
 in a ball of radius $R$.  In most cases the optimal balls will be centered at the origin,
 then $r$ and $R$ can be  equivalently defined as the ``best'' constants  for which
 $R^{-1}|\cdot| \leq \|\cdot\|_K \leq r^{-1}|\cdot|$.

If $u$ is a vector from the unit sphere $S^{m-1}$, the {\em support
function} of $K$ in the direction $u$ is $h_K(u):=\max _{x\in K}
\langle x, u\rangle = \|u\|_{K^\circ}$. Note that $h_K(u)$ is the
distance from the origin to the hyperplane tangent to $K$ in the
direction $u$. The {\em mean width}\footnote{It would have been
geometrically more precise to call this quantity the mean {\em
half}-width.} of $K$ is then defined as
\begin{equation}\label{mean:width:definition} w (K):=\int _{S^{m-1}}
h_K(u)\,d\s(u)=\int _{S^{m-1}} \|u\|_{K^\circ} d\s(u), 
\end{equation} where $\,d\s(u)$ is the normalized spherical
measure on the sphere $S^{m-1}$ (this definition makes sense for
any bounded set $K$).

The Urysohn's inequality (see, e.g., \cite{Pisier1989}) is a
fundamental result which compares the volume radius and the mean
width: for any convex body $K \subset \R^m$, we have
\begin{equation} \label{urysohn}
\vrad(K) \leq w(K) .
\end{equation}

\noindent It is often convenient to consider the Gaussian variant
of  mean width 
\[ w_G(K) := \E \|G\|_{K^\circ}, \]
where $G$ is a standard Gaussian vector in $\R^m$, i.e., a random
vector with independent $N(0,1)$ coordinates in any orthonormal
basis. One checks, by passing to polar coordinates, that for every
convex body $K \subset \R^m$
\begin{equation} \label{gammaN}
w_G(K) = \gamma_m \, w(K),
\end{equation}
where \be \label{gammaNbounds}
 \gamma_m := \E |G| = \frac{\sqrt{2}\Gamma((m+1)/2)}{\Gamma(m/2)},  \quad
\sqrt{m-1}\leq \gamma_m \leq \sqrt{m}, \ee
 is a constant depending only on $m$.

Note that we may extend the definition of the {\em Gaussian mean width} 
to all bounded sets $K \subset \R^m$ through the formula
\begin{equation} \label{defn-mean-width}
w_G(K) = \E \sup_{x \in K} \langle G,x \rangle. \end{equation} The
Gaussian mean width of $K$ is an intrinsic parameter and does not
depend on the ambient dimension: if $K$ lives in a subspace $E
\subset \R^m$, the formula \eqref{defn-mean-width}  gives the same
value whether $G$ is a standard Gaussian vector in $\R^m$, or a
standard Gaussian vector in $E$.

\subsection{Concentration of measure} \label{sec:concentration}

The phenomenon of {\em concentration of measure} 
plays a central role in our proofs.
We first recall the statement of L\'evy's lemma. In the statements
of Lemmas \ref{lemma:Levy} and \ref{lemma:Levy-local}, $\P$
 stands for the uniform measure on the sphere $S^{m-1}$,
normalized so that $\P(S^{m-1})=1$.

\begin{lemma}[L\'{e}vy's lemma \cite{Levy1951,MS}] \label{lemma:Levy}
 If $f: S^{m-1} \rightarrow \R$ is an $L$-Lipschitz
function, then for every $\e>0$,
\begin{equation*} \P(\left\{|f-M| >\e\right\}) \leq  C_1 \exp(-c_1m\e ^2/L^2),
\end{equation*}
 where $M$ is any central value of $f$, and
$C_1, c_1> 0$ are absolute constants.
\end{lemma}

By a {\em central value} of a random variable $X$ we mean either the
expectation or the median, or more generally  any number $M$ such
that $\P(X \geq M) \geq 1/4$ and $\P(X \leq M) \geq 1/4$. Any two
central values for the function $f$ appearing in L\'evy's lemma
differ by at most $C_2L/\sqrt{m}$.

While we will be {\em mostly} interested in concentration of functions 
on the sphere, the phenomenon appears also in many other contexts, 
for example in the Gaussian setting. A recent fairly comprehensive 
reference is the monograph \cite{ledoux}.

We will also consider situations in which a function $f$ has a
(possibly) large Lipschitz constant, while the restriction of $f$
to a large subset has a small Lipschitz constant. The following 
extension of L\'evy's lemma handles such a case. The trick behind
this lemma appeared in \cite{AubrunSzarekWerner2011} and
implicitly in \cite{Hasting2009}.

\begin{lemma}[L\'{e}vy's lemma, local version] \label{lemma:Levy-local}
Let $\Omega \subset S^{m-1}$ be a subset of measure larger than
$3/4$. Let $f : S^{m-1} \to \R$ be a function such that the
restriction of $f$ to $\Omega$ is $L$-Lipschitz. Then, for every
$\e>0$,
\begin{equation*}
\P(\{|f(x)-M_f| >\e\}) \leq  \P(S^{m-1} \setminus \Omega) + C_1
\exp(-c_1m\e ^2/L^2),\end{equation*} where $M_f$ is  the median of
$f$, and $C_1, c_1> 0$ are absolute constants.
\end{lemma}

In Lemma \ref{lemma:Levy-local}, the median can be replaced by
another quantile (up to changes in the numerical constants).
However, in general, it cannot be replaced by the mean (we do not
assume any regularity of $f$ outside $\Omega$, therefore the
expectation may even fail to be well-defined).  Still, more often
than not, {\em some} information about global regularity of  $f$
is available and concentration around the mean can be inferred;
see the comment at the end of the proof and the remark following
Lemma \ref{lem:lipschitz-estimate} below.

\begin{proof}[Proof of Lemma \ref{lemma:Levy-local}] 
The key point is that in any metric space $X$, it is possible to
extend any $L$-Lipschitz function $h$ defined on a subset $Y$
without increasing the Lipschitz constant. Use, e.g., the formula
\[ \tilde{h}(x) = \inf_{y \in Y} \left[ h(y) + L\, {\rm dist} (x,y) \right].\]

By applying this fact to $X=S^{m-1}$, $Y=\Omega$ and
$h=f_{|\Omega}$, we obtain a function $\tilde{f} : S^{m-1} \to \R$
which is $L$-Lipschitz and coincides with $f$ on $\Omega$.
Moreover, if $M=M_f$ is the median of $f$, then
\[ \P( \{\tilde{f} \geq M \}) \geq \P(\Omega \cap \{ \tilde{f} \geq M \})
= \P(\Omega \cap \{ f \geq M \} ) \geq 1/4. \] Similarly, $\P(
\{\tilde{f} \leq M \}) \leq 1/4$. Hence $M$ is a central value for
$\tilde{f}$. By L\'evy's lemma,
\[ \P(\{|\tilde{f}-M| >\e\}) \leq C_1 \exp(-c_1m\e ^2/L^2). \]
Therefore,
\begin{eqnarray*} 
\P(\{|f-M| >\e\}) & \leq & \P(\{f \neq \tilde{f}\}) + \P(\{|\tilde{f}-M| >\e\}) \\
& \leq & \P(S^{m-1} \setminus \Omega) + C_1 \exp(-c_1m\e ^2/L^2),
\end{eqnarray*}
as claimed. Note that if we know, for example, that $f$ is
``reasonably bounded,'' then we can infer that the median $M$ and
the mean $\E f$ of $f$ do not differ very much (e.g., $|M-\E f|
\lesssim \|f-M\|_\iy\,  \P(S^{m-1} \setminus \Omega)$) and deduce
{\em a posteriori} concentration around the mean.
\end{proof}

\subsection{Quantum states}

Throughout the paper, we consider a (finite dimensional) complex
Hilbert space $\cH$, equipped with a norm which we will also
denote by $|\cdot|$.

A {\em quantum state} on $\cH$ is a positive trace one operator on
$\cH$. We use $\cD=\cD(\cH)$ to denote the set of all quantum
states on $\cH$. The extreme points of this set are {\em pure} states,
in particular if $\cH=\bC^n$, then
\[ \cD(\C^n) = \conv \{ |\psi\rangle \langle \psi |  \st  \psi \in \C^n, |\psi|=1 \}. \]
Above (and, when convenient, in what follows) we use Dirac's
bra-ket notation: $|\psi\rangle$ is a column vector, $ \langle
\psi |  = |\psi\rangle^\dagger$ is a row vector and $|\psi\rangle
\langle \psi |$ is their outer product, the orthogonal projection
onto $\C  \psi$. The set $\cD(\C^n)$ is contained in the (real)
space  $\cM_n^{sa}$ of $n\times n$ self-adjoint matrices, endowed
with the Hilbert--Schmidt inner product $\langle A,B \rangle = \tr
(AB)$. Whenever considering a geometric invariant (inradius, mean
width, \ldots) of a set of matrices, it will be tacitly understood
that the Hilbert--Schmidt Euclidean structure is used. This
applies also to spaces of not-necessarily-self-adjoint and/or
rectangular matrices; in that case $\langle A,B \rangle = \tr
(AB^\dagger)$. We will also occasionally use the Schatten $p$-norm
$\| A \|_p = \big( \tr (A^\dagger A)^{p/2} \big)^{1/p}$. The limit
case $\|\cdot \|_\iy$ coincides with the operator norm
$\|\cdot\|_{op}$ (from the category of  normed spaces), while
 $\|\cdot\|_{2}=\|\cdot\|_{HS}$ is the Hilbert--Schmidt (or Frobenius) norm.

For every dimension $n$, we introduce now a family of probability
distributions on $\cD(\C^n)$ which plays a central role in this
paper. These probability measures are known as \emph{induced
measures} and can be described as follows. Fix a positive integer $s$ and
let  $|\psi \rangle \langle \psi |$ be a random pure state on the
Hilbert space $\C^n \otimes \C^s$, where $\psi$ is a random unit
vector  uniformly distributed on the sphere
 in $\C^n \otimes \C^s$. Then consider the partial trace
of $|\psi \rangle \langle \psi |$ over $\C^s$; the resulting state
is a random state on $\C^n$ and we denote by $\mu_{n,s}$ its
distribution.

When $s \geq n$, the probability measure $\mu_{n,s}$ has a density
with respect to the Lebesgue measure on $\cD(\C^n)$ which has a
simple form \cite{Zycz2}
\begin{equation} \label{eq:formula-density} \frac{d\mu_{n,s}}{d\! \vol}(\rho)
= \frac{1}{Z_{n,s}} (\det \rho)^{s-n} ,\end{equation} where
$Z_{n,s}$ is a normalization factor. Note that formula
\eqref{eq:formula-density} allows to define the measure
$\mu_{n,s}$ (in particular) for every real $s \geq n$, while the
partial trace construction makes sense only for integer values of
$s$.

In the important special case when $s=n$, the density of the
measure $\mu_{n,n}$ is constant. A random state distributed
according to $\mu_{n,n}$ is uniformly distributed on $\cD(\C^n)$
(i.e., uniformly with respect to the Lebesgue measure).

It is important to consider the case when the Hilbert space $\cH$
itself carries a tensor product structure. For simplicity we will
largely focus on the bipartite balanced case where $\cH = \C^d
\otimes \C^d$ (multipartite Hilbert spaces are considered in
Sections \ref{section:multipartite} and
\ref{section:alternative-argument}, while some remarks on
extensions to the unbalanced setting $\cH = \C^{d_1} \otimes
\C^{d_2}$, $d_1\neq d_2$ are given in Section \ref{unbalanced}).

A quantum state on $\C^d \otimes \C^d$ can be either separable or
entangled, and this dichotomy is fundamental in quantum theory. By
definition \cite{Wer1}, a state $\rho$ is {\em separable} if it can be
written as a convex combination of product states (i.e. states of
the form $\rho_1 \otimes \rho_2$, where $\rho_1,\rho_2$ are states
on $\C^d$). If we denote by $\cS(\C^d \otimes \C^d)$ $\subset
\cD(\C^d \otimes \C^d)$ the subset of separable states, an
equivalent description is the following
\[ \cS(\C^d \otimes \C^d)= 
\conv \{ | \psi_1 \otimes \psi_2 \rangle \langle \psi_1 \otimes \psi_2 | \st
\psi_1,\psi_2 \in \C^d, |\psi_1|=|\psi_2|=1 \}.\]

	A state which is not separable is called {\em entangled}. We denote
$n=d^2$ the (complex) dimension of the space $\C^d \otimes \C^d$,
which we identify with $\C^n$. The {\em affine} hyperplane
\[ \Hone =\{A\in \cM_{n}^{sa} \; : \;  \ \tr\,(A)=1\}, \]
contains the set $\cD = \cD(\C^n)$ of states and its subset $\cS=
\cS(\C^n)$ of separable states; they are both of full (real)
dimension $m=n^2-1$.

We want to consider $\Hone$ as a vector space where the role of
the origin is played by the {\em maximally mixed state} $\Id/n$, where
$\Id$ denotes the identity matrix. One way to formalize this point
of view is to work with the {\em linear} hyperplane
\[ \Hzero = \{A\in \cM_{n}^{sa}\  \; : \;  \ \tr\,(A)=0\} = \Hone - \Id/n \]
and with the translated   convex body
\[ \SEP = \cS - \Id/n = \{ \rho - \Id/n   \; : \;   \rho \in \cS \} \subset \Hzero.\]
Similarly, we denote $\cD_0=\cD-\Id/n$ (in the sequel, we will use
analogous notation also for other sets). The geometry of the sets
$\cS$ and $\cD$ plays a central role in our argument. Estimates on
some known geometric parameters associated to these convex bodies
are gathered in Table \ref{table-radii}.

\begin{table}[htbp]
\caption{Radii of $\cD$ and $\cS$ for $\C^d \otimes \C^d$, where
$n=d^2$. All these parameters are translation invariant, so their
values for $\cD_0$ and $\SEP$ are respectively the same. In each
row the quantities increase from left to right.}
\label{table-radii}
\begin{tabular}{|c|c|c|c|c|}
  \hline
   & inradius & volume radius & mean width & outradius\\
  \hline
  $\vphantom{\displaystyle \sum}  \cD(\C^d \!\otimes\! \C^d)$ &
  $= \!1/\!\sqrt{n(n\!-\!1)}$ &  $\simeq n^{-1/2}$ & $\simeq n^{-1/2}$ & 
  $=\! \sqrt{(n\!-\!1)/n}$ \\
  \hline
  $\vphantom{\displaystyle \sum} \cS(\C^d \!\otimes\! \C^d)$ &
  $= \!1/\!\sqrt{n(n\!-\!1)}$  & $\simeq n^{-3/4}$ & $\simeq n^{-3/4}$ & 
  $=\! \sqrt{(n\!-\!1)/n}$ \\
  \hline
\end{tabular}
\end{table}

The volume of $\cD$ was computed exactly in \cite{Zycz1} and it
was noted in \cite{Sz1} that the mean width and the volume radius
have the same order. The remarkable fact that $\cD$ and $\cS$ have
the same inradius was proved in \cite{BarnumGurvits2002}. (An
alternative argument is based on a dual formulation given in
\cite{SzarekWerner2008}, a proof of which was provided by H.-J.
Sommers, see \cite{OWR}.) Sharp bounds on the volume radius of
$\cS$ were given in \cite{AS1} (the ratio $\vrad(\cS)/\vrad(\cD)$
is estimated in Theorem 1 in \cite{AS1}). The estimate for the
mean width of $\cS$ does not appear explicitly in \cite{AS1},  but
follows from the argument since the upper bound on the volume
radius was obtained via Urysohn's inequality \eqref{urysohn}.
Finally, the calculation of the outradii is easy: they are
attained on pure states. Note that the inradii and the outradii of
$\cD$ and $\cS$ are attained on balls centered at $\Id/n$, which
is the only point invariant under isometries of each of these
bodies.

It is easily checked that $\cD_0(\C^n)^\circ = -n\cD_0(\C^n)$ 
(we recall that  the polar operation $^\circ$ 
is performed in the space $\Hzero$ of trace zero matrices).  
This is a consequence of the fact that the cone of positive
matrices is self-dual (cf. more general comments in the second paragraph
of Section \ref{section:estimate-s0}).
 We deduce immediately from Table \ref{table-radii}
that the mean width of $\cD_0(\C^n)^\circ$ is of order $\sqrt{n}$.
The situation is not so simple for $\cS$, and estimating the mean
width of $\SEP(\C^d \otimes \C^d)^\circ$ is the main technical
difficulty in our argument.

\subsection{Threshold for entanglement vs separability}
\label{results}

The main result of the paper is the following theorem.

\begin{theorem}
\label{threshold:separability} There are effectively computable
absolute constants $C,c>0$ and a function $s_0(d)$ satisfying
\[ cd^{3} \leq s_0(d) \leq Cd^{3} \log^2 d ,  \]
such that if $\rho$ is a random state on $\C^d \otimes \C^d$
distributed according to the measure $\mu_{d^2,s}$, then, for any
$\e>0$,
\begin{itemize}
 \item [(i)] if $s \leq (1-\e)s_0(d)$, we have
 \[ \P(\rho \textnormal{ is separable}) \leq 2  \exp(-c(\e)d^3), \]
 \item [(ii)] if $s \geq (1+\e)s_0(d)$ we have
 \[ \P(\rho \textnormal{ is entangled}) \leq 2 \exp(-c(\e)s ), \]
\end{itemize}
where $c(\e)$ is a positive constant depending only on $\e$.
\end{theorem}

Theorem \ref{threshold:separability} asserts that, for fixed $d$,
the character of the induced state changes sharply from generic
entanglement to generic separability as $s$, the dimension of the
ancilla, increases. If we knew that the threshold function
$s_0(\cdot)$ was regular enough, an analogous statement with the
roles of $d$ and $s$ exchanged would immediately follow. 
The following is a result
in that direction which can be deduced with relatively little
effort. (A statement in a similar language -- but much less precise
--- appears in \cite{KZM}.)

\begin{corollary} \label{particles}
Consider a system of $N$ identical particles (qudits)  in a random
pure state. Then there is a threshold $k_0 = k_0(N)\sim N/5$ such
that
 two subsystems of $k$ particles each typically share entanglement if
$k>k_0$,  and  typically do not share entanglement if $k<k_0$.
\end{corollary}

We next describe the threshold function $s_0(d)$ appearing in the
main theorem. Let $G$ denote the standard Gaussian vector in the
space $\Hzero$, which we will also call a $\text{GUE}^0$ random
matrix. One may represent $G$ as follows. Start from an $n \times
n$ $\text{GUE}$ random matrix $G'$, which is the standard Gaussian
vector in the space $\cM_{n}^{sa}$ (see \cite{agz,ds-handbook}).
Then $G$ can be realized as $G'-\frac{\tr G'}{n} \Id$  (a
conditional expectation of $G'$). Equivalently, one may realize
$G$ by conditioning $G'$ to be of trace $0$. In this notation,
$s_0$ is defined as follows.

\begin{definition}
For every integer $d$, we define $s_0=s_0(d)$ by the formula
\begin{equation} \label{eq:def-s0} s_0(d) := \left(\frac{\E \|G\|_{\SEP}}{d^2} \right)^2
\sim w(\SEP^\circ)^2,\end{equation} where $G$ is a
$\textnormal{GUE}^0$ matrix of size $d^2 \times d^2$. 
(The relation ``$\sim$'' is  justified by \eqref{gammaN} and \eqref{gammaNbounds}.)
\end{definition}

\subsection{Overview of the proof} \label{overview}

Our proof of Theorem \ref{threshold:separability} consists of two
largely independent parts
\begin{itemize}
 \item showing that $s_0$ defined by \eqref{eq:def-s0} 
 is indeed a sharp threshold for the separability of random states,
 \item proving that $d^3 \lesssim s_0 \lesssim d^3 \log^2 d$.
\end{itemize}
The details of the two parts will 
be dealt with  
in Sections  \ref{section:threshold} and  \ref{section:estimate-s0} respectively 
(except for some fine points regarding the probability estimate in
part (i) of the Theorem, which are clarified in Section \ref{sharpness}). 
To not to obscure the structure of the proof, some general (but 
technically involved) auxiliary results are relegated to appendices. 
The heuristic behind deducing Corollary \ref{particles} from Theorem
\ref{threshold:separability} was explained in the Introduction. A
rigorous argument requires two additional simple observations
(Lemmas \ref{lem:monotonicity} and \ref{lem:monotonicity2}) and is
sketched in Section \ref{proof:particles}.

And here is a ``high level''  overview of the argument. 
We first note tautological equivalences

\smallskip 
\quad $\rho$ is separable  $\iff$ $\rho \in \cS$  $\iff$ $\rho - \Id/n \in \SEP$  $\iff$ 
$\|\rho - \Id/n\|_{\SEP} \leq 1$.

\smallskip \noindent 
This means that Theorem \ref{threshold:separability} asserts that 
if $s$ is {\em noticeably  larger} than $s_0$ (that is, if $s \geq (1+\e)s_0$),  
then $\|\rho - \Id/n\|_{\SEP} \leq 1$ with probability close to $1$, and if 
$s$ is  noticeably  {\em smaller} than $s_0$, then  $\|\rho - \Id/n\|_{\SEP} > 1$ 
with probability close to $1$.  

The strategy is now to show that the function  $\rho \to \|\rho - \Id/n\|_{\SEP}$ 
is sufficiently regular 
(which is relatively straightforward) and that whenever $s$ is noticeably  larger 
than $s_0$,  then the median (or the expected value) of $\|\rho - \Id/n\|_{\SEP}$ 
is noticeably smaller than $1$.  The concentration phenomenon 
(Lemmas \ref{lemma:Levy} and \ref{lemma:Levy-local}) then implies that 
$\|\rho - \Id/n\|_{\SEP} \leq 1$ with probability close to $1$, as needed. 
Similarly, if $s$ is  noticeably  {\em smaller} than $s_0$, 
we need to establish that  the expected value of $\|\rho - \Id/n\|_{\SEP}$ 
is noticeably larger than $1$.

The problem thus reduces to figuring out the dependence of 
the expected value of $\|\rho - \Id/n\|_{\SEP}$ on the ancilla dimension $s$, 
which is implicit in the definition of the random state $\rho=\rho_{n,s} 
=\tr_{\C^s} | \psi \rangle \langle \psi |$.  
This turns out to be not so easy, partly because of the non-linear dependence 
of $\rho$ on $\psi$
(a random vector uniformly distributed on the unit sphere in $\C^n
\otimes \C^s$). However, it turns out (see Proposition \ref{prop:MP-vs-GUE} below) 
that, for all practical purposes, 
$\rho_{n,s} - \Id/n$ is equivalent to $A_{n,s}:= \frac{1}{n\sqrt{s}}G$, where 
$G$ is a random matrix distributed according to the standard Gaussian 
measure in the space of  $n\times n$ self-adjoint matrices 
with vanishing trace ($\text{GUE}^0$,
defined above in Section \ref{results}).  

This simplifies matters significantly since, first,  the dependence of 
$A_{n,s}$ on $s$ is {\em very} straightforward and, second, 
because the expected value of  $\|G\|_{\SEP}$ has geometric meaning: 
as explained in Section \ref{geometry}, it is explicitly related to 
$w(\SEP^\circ)$, the mean width of the polar of $\SEP$. 
However, estimating $w(\SEP^\circ)$ directly is still hard. 
The approach which succeeds is to proceed through a duality
argument, the idea being that for a  
``well-balanced''  convex body $K$,  its mean width $w(K)$ 
and the mean width of its polar, $w(K^\circ)$,  are approximately reciprocal, 
and that good estimates for $w(\SEP)$ exist in the literature. 

Some of the steps indicated above (for example, showing 
that the set $\SEP$ is well-balanced in the needed sense) 
are quite involved by themselves; we will try to convey 
additional heuristic arguments  clarifying such steps once  
precise statements are formulated and once appropriate 
notation is available.

\section{Proof that $s_0$ is a threshold for separability} 
\label{section:threshold}

Let $\rho = \rho_{n,s}$ be a random state on $\C^n$ with
distribution $\mu_{n,s}$. The first step is to approximate
$\rho-\Id/n$ by $\frac{1}{n\sqrt{s}}G$, where $G=G_n$ is an 
$n\times n$ $\text{GUE}^0$ random matrix.  
(This step would 
not be necessary if we were able to define the threshold dimension $s_0$ 
via the expected value -- or median -- of 
$\| \rho- {\Id}/{n}\|_{\SEP}$ rather than in terms of 
 the expected value of $\| G\|_{\SEP}$.)  
Here is a precise statement.

\begin{proposition} \label{prop:MP-vs-GUE}
Denote by $c_{n,s}$ (resp. $C_{n,s}$) the largest (resp. smallest)
constant such that for every convex body $K \subset \cM^{sa,0}_n$
containing $0$ in its interior, if $\rho$ is a random state on
$\C^n$ distributed according to $\mu_{n,s}$, and if $G$ is a
standard Gaussian vector in $\cM^{sa,0}_n$, we have
\begin{equation}
\label{eq:Cns} \frac{c_{n,s}}{n\sqrt{s}} \E \| G\|_{K} \leq \E
\left\| \rho- \frac{\Id}{n} \right\|_K \leq
\frac{C_{n,s}}{n\sqrt{s}} \E \| G\|_{K} .\end{equation} Then
 \[ \lim_{n,\frac{s}{n} \to \iy} C_{n,s} = \lim_{n,\frac{s}{n} \to \iy} c_{n,s} = 1 .\]
\end{proposition}
For the record, 
let us clarify what double-indexed limits mean (here and later).
The statement
\[ \lim_{n,\frac{s}{n} \to \iy} C_{n,s} =1 \]
is supposed to signify the following: for any sequences $(n_k), (s_k)$
such that both $(n_k)$ and $(s_k/n_k)$ tend to infinity, we have $
\displaystyle  \lim_{k \to \iy} C_{{n_k},{s_k}} =1$.

A rigorous proof of Proposition \ref{prop:MP-vs-GUE} is given in Appendix
\ref{app:marcenko-pastur-vs-GUE}; here we restrict ourselves to  
some heuristic comments.  First, both $\rho = \rho_{n,s}$ and $G=G_n$ are 
well-known ensembles in  Random Matrix Theory. They are both 
invariant under conjugation with a unitary matrix, and their 
asymptotic spectral properties have been thoroughly studied. 

The behavior of $G_n$ for large $n$ is governed by the famous 
{\em Wigner's semi-circle law}. 
On the other hand, (appropriately normalized) $\rho_{n,s}$  
is known as the Wishart ensemble and, when $n,s \to \infty$ with the ratio 
$s/n \to \beta$ for some $\beta >0$, the limiting spectral distribution 
 is given by the  {\em Marchenko--Pastur law}. 
 However, in the asymptotic regime that is relevant here 
($n,s/n \to \infty$),   the limiting spectral distribution is also a 
(non-centered, that's why we subtract ${\Id}/{n}$) semi-circle law. 

Having noticed that 
$\rho - \Id/n$ and $\frac{1}{n\sqrt{s}}G$ 
have the same limiting spectral distribution, 
we need to deduce that this implies their asymptotic equivalence in the sense of 
\eqref{eq:Cns}.  This is done in two steps.
 First, we point out that  known results about convergence 
 to the semicircle law can be subsumed  in the language 
 of the so-called $\infty$-{\em Wasserstein distance}  
  (in random matrix theory, such results are usually stated in a rather weak form).  
 Next we show that this (combined with unitary invariance) 
 implies that the expectations of the gauges
$\|\cdot\|_K$ must be asymptotically the same for both ensembles; 
this part of the argument is based on  Appendix \ref{app:wasserstein} 
and on the so-called {\em majorization theory}. 
We emphasize that the latter step is delicate since there are no 
uniform assumptions on continuity of the gauge $\|\cdot\|_K$.

\begin{remark}
While the formulation of Proposition \ref{prop:MP-vs-GUE} focuses
on the regime when $n$ and $s/n$ tend to infinity, the proof can
be adapted to other situations. For example, one can show that,
for any $\alpha>0$,
\[ 0 < \inf_{s \geq \alpha n} c_{n,s} \leq \sup_{s \geq \alpha n} C_{n,s}< +\iy . \]
This allows to establish a threshold phenomenon even for
properties
 -- in place of separability -- for which
$s_0(d)  \simeq d^2$. However, in that case the argument does not
yield the sharp threshold property, i.e., involving arbitrary $\e
> 0$.  See Section \ref{other} for more comments on related
issues.
 \end{remark}

We now return to the proof of assertions (i) and (ii) of 
Theorem \ref{threshold:separability}. 
Applying Proposition \ref{prop:MP-vs-GUE} for $K=\SEP$ and using
the definition \eqref{eq:def-s0} of $s_0(d)$, we obtain that (when
$d$ and $s/d^2$ tend to infinity)
\begin{equation}
\label{MPorder} \E \left\| \rho- \frac{\Id}{n} \right\|_{\SEP}
\sim \sqrt{\frac{s_0(d)}{s}} .
\end{equation}
Since a state $\rho$ is separable when $\|\rho-\Id/n\|_{\SEP} \leq
1$ and entangled when $\|\rho-\Id/n\|_{\SEP} > 1$, this suggests
that separability is typical when $s > s_0(d)$ and entanglement is
typical when $s < s_0(d)$. This will be made rigorous through the
next proposition; its proof is based on concentration of measure 
(reviewed in Section \ref{sec:concentration}).

\begin{proposition} \label{prop:concentration}
Let $s\geq n$, let $K \subset \cD(\C^n)$ be a convex body with
inradius $r$, and let $\rho$ be a random state with distribution
$\mu_{n,s}$. Let $M$ be the median of $\|\rho-\Id/n\|_{K_0}$, with
$K_0=K-\Id/n$. Then, for every $\eta>0$,
\[ \P \left( \left| \Big\|\rho - \frac{\Id}{n} \Big\|_{K_0} -M \right|
\geq \eta \right) \lesssim \exp(-cs)+\exp(-cn^2sr^2\eta^2) .\]
\end{proposition}

\begin{proof}[Proof]
Let $\rho$ be a random state on $\C^n$ with distribution
$\mu_{n,s}$. By definition, $\rho$ has the same distribution as
\[ \tr_{\C^s} | \psi \rangle \langle \psi |, \]
where $\psi$ is uniformly distributed on the unit sphere in $\C^n
\otimes \C^s$. Equivalently, $\rho$ has the same distribution as
$AA^\dagger$, where $A$ is an $n \times s$ matrix uniformly
distributed on the Hilbert--Schmidt sphere $S_{HS}$ (this is not 
immediately obvious, but can be verified by a straightforward 
calculation; also,  $S_{HS}$ can be
identified with the real sphere $S^{2ns-1}$). Consider the
function $f : S_{HS} \to \R$  defined by
\begin{equation} \label{eq:define-function} 
f(A) = \left\| AA^\dagger - \frac{\Id}{n} \right\|_{K} .
\end{equation}

\begin{lemma} \label{lem:lipschitz-estimate}
For every $t>0$, denote by $\Omega_t$ the subset
\[ \Omega_t = \{ A \in S_{HS} \st \|A\|_{\iy} \leq t \} .\]
Then the Lipschitz constant of the restriction of $f$ to
$\Omega_t$ is bounded by $2t/r$.
\end{lemma}

\begin{remark}
In particular, taking $t=1$, one obtains that the global Lipschitz
constant of $f$ is bounded by $2/r$. This implies that any two
central values for $f$ differ by at most $C/(r\sqrt{ns})$.
\end{remark}

\begin{proof}
The function $f$ is the composition of several operations:
\begin{itemize}
 \item the map $A \mapsto \|A\|_{K}$, which is $1/r$-Lipschitz with respect to
 the Hilbert--Schmidt norm.
 \item the map $A \mapsto A-\Id/n$, which is an isometry for the Hilbert--Schmidt norm,
 \item the map $A \mapsto AA^\dagger$, whose Lipschitz constant can be estimated
 by the following chain of inequalities
\end{itemize}
\begin{eqnarray}
\nonumber \|AA^\dagger-BB^\dagger\|_{2} &
\leq & \| A(A^\dagger-B^\dagger) + (A-B)B^\dagger \|_{2} \\
\label{eq:lipschitz-constant} &
\leq & \| A \|_{\iy} \| A^\dagger - B^\dagger \|_{2} + \|A-B\|_{2} \|B^{\dagger}\|_{\iy} \\
\nonumber & \leq & (\|A\|_{\iy} + \|B\|_{\iy}) \|A-B\|_{2} \, .
\end{eqnarray}
In particular, if $A,B \in \Omega$, we obtain $\|AA^\dagger -
BB^\dagger\|_{2} \leq 2t \|A-B\|_{2}$.
\end{proof}

We now apply Lemma \ref{lemma:Levy-local}, and we obtain that for
every $\eta>0$,
\[ \P( |f-M| \geq \eta) \lesssim \P(S_{HS} \setminus \Omega_t)
+ \exp(-c_1ns\eta^2 (2t/r)^{-2}) .\]

If we choose $t=3/\sqrt{n}$, then $\P(S_{HS} \setminus \Omega_t)
\lesssim \exp (-cs)$ (this follows from an elementary net argument, as
explained in Lemma 6  and  Appendix B of \cite{AubrunSzarekWerner2011}) 
and the result follows. 
\end{proof}

We shall now show how to use Propositions \ref{prop:MP-vs-GUE}
and \ref{prop:concentration}  to conclude the proof  of assertions 
(i) and (ii) of  Theorem
\ref{threshold:separability}. For any integers $d$ and $s$, we
define $\pi_{d,s}$ as the probability that a random state on $\C^d
\otimes \C^d$ with distribution $\mu_{d^2,s}$ is separable. We
first show that $\pi_{d,s}$ is decreasing with respect to $d$.

\begin{lemma} \label{lem:monotonicity}
Let $s,d_1,d_2$ be integers such that $d_1 \leq d_2$. Then
\[ \pi_{d_2,s} \leq \pi_{d_1,s}. \]
\end{lemma}

\begin{proof}
Identify $\C^{d_1}$ as a subspace of $\C^{d_2}$, and let $Q :
\C^{d_2} \to \C^{d_1}$ be the orthogonal projection. Then,
$\C^{d_1} \otimes \C^{d_1} \subset \C^{d_2} \otimes \C^{d_2}$ is
the range of the projection $P = Q \otimes Q$.
 Let $\rho_2$ be a random state on $\C^{d_2} \otimes \C^{d_2}$
with distribution $\mu_{d_2^2,s}$. Then the state
\begin{equation} \label{eq:coupling} 
\rho_1 := \frac{P \rho_2 P}{\tr P \rho_2 P} 
\end{equation}
is a random state on $\C^{d_1} \otimes \C^{d_1}$ with distribution
$\mu_{d_1^2,s}$ (this is obvious if we realize $\rho_2$ as
${GG^\dagger}/{\tr GG^\dagger}$, where $G$ is a $d_2^2 \times s$
random matrix with i.i.d. $N_\C(0,1)$ entries). The inequality
follows from the fact that, given the relation
\eqref{eq:coupling}, the separability of $\rho_2$ implies the
separability of $\rho_1$ (local operations cannot create
entanglement).
\end{proof}

\begin{remark}
Another natural problem is whether $s_1 \leq s_2$ implies
$\pi_{d,s_1} \leq \pi_{d,s_2}$, i.e. whether the probability that
a random induced state is separable always increases with the
dimension of the environment. We do not know the answer to this
question (see Section \ref{sharpness}).
\end{remark}

\begin{proof}[Proof  of assertions (i) and (ii) of  Theorem \ref{threshold:separability}]
We first address part (i): fix \,$\e\hskip-.5mm>~\hskip-.5mm0$, then for any $s$ and $d$
satisfying the condition $s \leq (1-\e)s_0(d)$, we have to show
that
\begin{equation} \label{dcube}
\pi_{d,s} \leq 2 \exp(-c(\e) d^3 ).
\end{equation}
We start by establishing a slightly different estimate
\begin{equation} \label{sbound} \pi_{d,s} \leq 2  \exp(-c(\e)s),\end{equation}
which is stronger than \eqref{dcube} in the crucial range  $s
\gtrsim d^3$. The case $s =o(d^3)$ of  \eqref{dcube}  can be then
deduced formally using Lemma  \ref{lem:monotonicity} and other
known facts. We relegate the details to Section \ref{sharpness}
since  the qualitative information provided by  \eqref{dcube} and
\eqref{sbound} is essentially the same; the only reason why we
{\sl did not} state the Theorem with the bound  \eqref{sbound} is
that it would give a misleading impression that the genericity of
entanglement wanes as $s$ decreases, while in fact the opposite is
true.

In view of Lemma \ref{lem:monotonicity}, to prove \eqref{sbound}
it is enough to consider the case when, for a given $s$, $d=d_s$
is the minimal integer satisfying $s \leq (1-\e)s_0(d)$. In
particular, we have then
\[ \lim_{s \to \iy} s/d^2 = +\iy \]
and we are in the asymptotic regime described in Proposition
\ref{prop:MP-vs-GUE}. Let $M_s$ (resp. $E_s$) be the median (resp.
the expectation) of $\| \rho - \Id/d^2 \|_{\SEP}$, when $\rho$ is
a random state on $\C^{d} \otimes \C^{d}$ with distribution
$\rho_{d^2,s}$. Applying Proposition \ref{prop:MP-vs-GUE}, we
have, when $s \to \iy$,
\[ E_s \sim \sqrt{\frac{s_0(d)}{s}} \sim \frac{1}{\sqrt{1-\e}}.  \]

We now apply Proposition \ref{prop:concentration} to the convex
body $K=\SEP$ (its inradius is of order $1/n$, cf. Table
\ref{table-radii}). This gives, for any $\eta>0$,
\begin{equation} \label{concentration}
\P \left( \left| \|\rho - \Id/d^2 \|_{\SEP} -M_s \right| \geq \eta
\right) \lesssim \exp(-cs)+\exp(-cs\eta^2) .
\end{equation}

Moreover, we have $|M_s - E_s| \lesssim d/\sqrt{s}$ (cf. the
remark following Lemma \ref{lem:lipschitz-estimate}). Since
$d/\sqrt{s}$ tends to $0$, $M_s$ is also equivalent to
$1/\sqrt{1-\e}$. We now choose $\eta>0$ such that
$1/\sqrt{1-\e}-\eta>1$, and we obtain that for $s$ large enough
\[ \pi_{d,s}  =  \P \left( \|\rho - \Id/d^2 \|_{\SEP} \leq 1 \right)
 \leq  C\exp(-cs)+C\exp(-cs\eta^2). \]
Hence there are  constants $C',c'(\e)>0$ such, for every $s$, 
we have  the inequality $
\pi_{d,s} $ 
$ \leq C' \exp(-c'(\e)s)$  (small values
of $s$ are taken into account by adjusting the constants if
necessary). A priori we may have $C'>2$, but in that case the bound
\eqref{sbound}  
follows with $c(\e) = c'(\e)/\log_2 C'$. This  shows part (i) of
Theorem \ref{threshold:separability}, except for the fine points
related to the difference between the tail estimates \eqref{dcube}
and  \eqref{sbound}, which will be clarified in Section
\ref{sharpness}. The part (ii) is proved in the same way as
\eqref{sbound}.
\end{proof}

\section{Estimation of $s_0$} \label{section:estimate-s0}

This section is devoted to the proof of the inequalities
\[ d^3 \lesssim s_0(d) \lesssim d^3 \log^2 d ,\]
which comprise the first assertion of Theorem \ref{threshold:separability}. 
By formula \eqref{eq:def-s0} defining $s_0$, these inequalities are equivalent to
\begin{equation} \label{inequalities-s0-2} 
d^{7/2} \lesssim \E \|G\|_{\SEP} \lesssim d^{7/2} \log d .
\end{equation}
That is,  we need to estimate  
the Gaussian average value of the gauge $\|\cdot\|_{\SEP}$.

It turns out that evaluating (or even estimating)  $\E \|G\|_{\SEP}$ {\em directly} 
is not easy. This may conceivably be related to the
fact that computing $\|\cdot\|_{\SEP}$ is an NP-hard problem
\cite{Gurvits-NPHard}. Alternatively, we may note that $\E
\|G\|_{\SEP}$ is directly related (via \eqref{gammaN}) to the mean
width of $\SEP^\circ$. Since there is a canonical link between
duality of cones and duality of bases of cones (see
\cite{SzarekWerner2008}, Lemma 1), it follows that any question
about $\SEP^\circ$ is equivalent to a question about the cone of
block-positive matrices and --  via the Choi--Jamio{\l}kowski
isomorphism -- to a question about the notoriously difficult to
study cone of positivity-preserving maps on $\cM_d$, the algebra
of $d\times d$ complex matrices (see \cite{Ben1}, or sections II
and III in \cite{SzarekWerner2008} for details).

The approach which succeeds is to proceed through a duality
argument. First, we estimate $\E \|G\|_{\SEP^\circ}$  
(or, equivalently, the mean width of $\SEP$), which is an
easier task. Then, we use a general theorem saying that for any
``well-balanced'' symmetric convex body  one can deduce the
average of the norm from the average of the dual norm, with a
multiplicative error logarithmic in the dimension. Since these 
 aspects of the theory of high-dimensional convex geometry
require the hypothesis of symmetry, we intoduce
the following {\em symmetrization} of the convex body $\SEP$
\[ \K = - \SEP \cap \SEP .\]
 
We first check that the relevant geometric parameters
are essentially unchanged by this symmetrization procedure.

\begin{proposition} \label{prop:symmetrization}
The convex bodies $\SEP$ and $\K$ have
\begin{enumerate}
\item[(i)] comparable average gauge:
\[
 \E \|G\|_{\SEP} \leq \E \|G\|_{\K} \leq 2\E \|G\|_{\SEP},
\]
\item[(ii)] comparable volume radius:
\[ \frac{1}{2} \vrad(\SEP) \leq \vrad(\K) \leq \vrad(\SEP), \]
\item[(iii)] comparable mean width:
\[ w(\SEP) \simeq w(\K) \simeq n^{-3/4} , \]
\item[(iv)] the same inradius, equal to $(n(n-1))^{-1/2}$. 
However, the outradius of $\K$ is bounded by $1/\sqrt{n}$, 
while the outradius of $\SEP$ is of order $1$.
\end{enumerate}
\end{proposition}

\begin{proof}
We have
\[ \|A\|_{\K} = \max (\|A\|_{\SEP},\|-A\|_{\SEP})  \leq \|A\|_{\SEP} + \|-A\|_{\SEP},\]
from which (i) follows, because the distribution of $G$ is symmetric.

The fact that volume is preserved is less elementary. 
Several results in this direction are due to Rogers--Shephard 
(cf. \cite{Rogers1957});  
they additionally assert that the worst case occurs when the body is a simplex. 
For the present symmetrization, we use 
the following inequality which is a variation on the
Rogers--Shephard inequality
\begin{proposition} \label{prop:rogers-shephard}
If $K \subset \R^m$ is a convex body with center of mass at the
origin, then
\[ \vol (-K \cap K) \geq 2^{-m} \vol (K). \]
\end{proposition}
A proof can be found in \cite{MilmanPajor} (Corollary 3). The
factor $2^{-m}$ is not likely to be sharp; it is again tempting to 
conjecture that  the simplex is the extremal case, but this seems to be unknown. 

We apply Proposition \ref{prop:rogers-shephard} to
$\SEP$ (to check that $0$ is the center of mass of $\SEP$, average over local unitaries) and
obtain
\[ \vrad( \K ) \geq \frac{1}{2} \vrad(\SEP), \]
which shows (ii) (the other inequality is trivial).

For (iii), we already know (cf. Table \ref{table-radii}) that
\[ \vrad(\SEP) \simeq w(\SEP) \simeq n^{-3/4}. \]
We therefore have the following chain of inequalities (the first is trivial, the 
third is (ii) and the last is Urysohn's  inequality \eqref{urysohn})
\[ w(\K) \leq w(\SEP) \simeq \vrad(\SEP) \simeq \vrad(\K) \leq w(\K).\]
Therefore all these quantities are comparable, and (iii) follows.

For (iv), the statement about inradius is trivial. On the other hand, any matrix $A \in \SEP$
satisfies $A \geq -\Id/n$. This implies that any $A \in \K$
satisfies $-\Id/n \leq A \leq \Id/n$, or  $\|A\|_{\iy} \leq 1/n$,
and therefore the outradius of $\K$ is bounded by $1/\sqrt{n}$. 
This completes the proof of Proposition \ref {prop:symmetrization}. 
\end{proof}

The required estimates for $s_0$ follow now from the next lemma

\begin{lemma} \label{lem:product:mean:width}
In the notation of the present section, we have
\[ n^2 \lesssim \E \|G\|_{\K} \cdot \E \|G\|_{\K^\circ} \lesssim n^2 \log n.\]
\end{lemma}

Indeed, in view of  \eqref{gammaN} and \eqref{gammaNbounds}, Proposition 
\ref{prop:symmetrization}(iii) implies that 
\begin{equation} \label{eq:mean-width-sym} \E \|G\|_{\K^\circ} \simeq n^{1/4}. 
\end{equation} 
From Lemma \ref{lem:product:mean:width}, we infer that
\[ n^{7/4} \lesssim \E \|G\|_{\K} \lesssim n^{7/4} \log n,   \]
and the inequalities \eqref{inequalities-s0-2} follow from part (i) 
of Proposition \ref{prop:symmetrization} (recall that
$n=d^2$, hence $\log n = 2 \log d$).  Since \eqref{inequalities-s0-2} 
was equivalent to the first assertion of Theorem  \ref{threshold:separability}, 
to conclude the proof of the Theorem it remains to show 
Lemma \ref{lem:product:mean:width}. 

\begin{proof}[Proof of Lemma \ref{lem:product:mean:width}]
The proof makes use of the $\ell$-position of convex bodies, which
is reviewed in Appendix \ref{app:ell-position}; here we just mention 
that a body is in the $\ell$-position if it is isotropic is some precise 
technical sense. (Note, however, that there are different notions of 
isotropy and the $\ell$-position is not the most common one.) 
We also point out that the
first inequality in Lemma \ref{lem:product:mean:width} is
elementary, see the last paragraph in Appendix
\ref{app:ell-position}.

The upper inequality will follow from the $MM^*$-estimate (Proposition
\ref{thm:MMstar}), which
is valid for any symmetric convex body that is in the
$\ell$-position. If we knew that the convex body $\K$ was in the
$\ell$-position, the result would be just an instance of the 
$MM^*$-estimate  (applied with $m=n^2-1$, which implies $\log m <
2 \log n$). However, there are not enough symmetries present to
conclude automatically that $\K$ is in the $\ell$-position.

We proceed as follows. Let $E \subset \Hzero$ be the subspace
spanned by the operators of the form $\sigma_1 \otimes \sigma_2$,
where $\sigma_1$ and $\sigma_2$ are self-adjoint operators with
trace $0$ on $\C^d$. Let  $F$ be the orthogonal complement of $E$
in $\Hzero$. We then have
\[ F = \{ \sigma_1 \otimes \Id \st  \tr \sigma_1 = 0 \} 
\oplus \{ \Id \otimes \sigma_2 \st \tr \sigma_2 = 0 \}
=: F_1 \oplus F_2. \] Clearly $\dim E=(n-1)^2$ and $\dim F=2n-2$.

Let $u:\Hzero \to \Hzero$ be a linear map such that $u(\K)$ is in
the $\ell$-position. By combining Lemma \ref{lem:irreducible} with
Lemma \ref{lem:unique-ell} from the Appendices, we may assume that
$u$ has the form
\[ u = P_{E} + ({\bf 0}_E \oplus v) \]
for some (positive definite operator)  $v : F \to F$, where ${\bf
0}_E$ is the zero operator on $E$. The ideal property of the
$\ell$-norm implies that
\[ \ell_{\K}(P_E)=\ell_{\K}(uP_E) \leq \ell_{\K}(u) ,\]
and similarly for $\ell_{\K^\circ}(P_E)$. By the $MM^*$-estimate 
(Proposition \ref{thm:MMstar}), 
we know that 
\[ \ell_{\K}(u)\ell_{\K^\circ}(u^{-1}) \lesssim n^2 \log n,\] 
and therefore $\ell_{\K}(P_E)\ell_{\K^\circ}(P_E)
\lesssim n^2 \log n$ (note that $u^{-1} = P_E + \big({\bf 0}_E
\oplus v^{-1}\big)$). To deduce similar estimates for $\Id$ in
place of $P_E$ we need the following observation.

\begin{claim} \label{claim}
$\ell_{\K}(P_{F})=o(\ell_{\K}(\Id))$ and
$\ell_{\K^\circ}(P_{F})=o(\ell_{\K^\circ}(\Id))$.
\end{claim}

Once the claim is proved, using the triangle inequality to bound
$\ell_{\K}(\Id) \leq \ell_{\K}(P_E) + \ell_{\K}(P_{F})$ (and
similarly for $\ell_{\K^\circ}$), we obtain
$\ell_{\K}(\Id)\ell_{\K^\circ}(\Id) \lesssim n^2 \log n$, which is
equivalent to the second inequality in
 the assertion of Lemma \ref{lem:product:mean:width}.
\end{proof}

\begin{proof}[Proof of Claim \ref{claim}] 
We use the estimates on the inradius and the outradius of $\K$ 
(see Proposition \ref{prop:symmetrization}(iv)) to deduce
the following inequalities (recall that $\gamma_m$ 
was defined in \eqref{gammaNbounds})
\[ \ell_{\K}(P_{F}) = \E \|P_{F}G\|_{\K} = 
w_G({(\K \cap F)^\circ}) \leq n \gamma_{\dim F} \lesssim n^{3/2},
\]
\[ \ell_{\K^\circ}(P_{F}) = \E \|P_{F}G\|_{\K^\circ} = w_G(P_F\K)\leq
 n^{-1/2} \gamma_{\dim F} \lesssim 1.
\]
On the other hand, we have $\ell_{\K^\circ}(\Id) \simeq n^{1/4}$
(by equation \ref{eq:mean-width-sym}) and $\ell_{\K}(\Id) \gtrsim
n^{7/4}$ (by the already shown lower estimate from Lemma
\ref{lem:product:mean:width}). This proves Claim \ref{claim} and
concludes the proof of Lemma \ref{lem:product:mean:width}, and hence of
the first assertion of Theorem \ref{threshold:separability}.  
Combined with the arguments in the preceding section,  
this concludes the proof of the Theorem.
\end{proof}
\begin{remark}
In the proof of Claim \ref{claim}, we upper-bounded the mean
widths of the convex bodies $P_{F}K$ and $(K \cap F)^\circ$ by
their outradii. This is sufficient for the present argument, but
these estimates are far from optimal. A more refined analysis is
performed in Lemma \ref{lem:section-of-SEP}; it will be needed to
handle the case of multipartite systems and unbalanced bipartite
systems.
\end{remark}

\section{Extension to the multipartite case} \label{section:multipartite}

 In this section we consider the case of a multipartite system,
 on the Hilbert space $\cH = (\C^d)^{\otimes k}$,
and estimate the threshold for separability of random states in
the asymptotic regime when $k\geq 2$ is fixed and $d$ tends to
$+\iy$. We denote $n=d^k = \dim \cH$ and $m=n^2-1= \dim \cD(\cH)$.
In this section, constants are allowed to depend on $k$; this is
emphasized by writing $\lesssim_k$,  $\simeq_k$,  $o_k(\cdot)$
instead of $\lesssim$, $\simeq$, $o(\cdot)$.

The set of separable states on $\cH$ is the subset of $\cD(\cH)$
defined as
\[ \cS(\cH) \!=\! \conv \{ | \psi_1 \otimes \cdots \otimes \psi_k \rangle \langle \psi_1
\otimes \cdots \otimes \psi_k | : \ \psi_i \in \C^d,
|\psi_i|=1, 
i=1, \cdots k \}.\]
As in the bipartite case, our argument requires estimates on
geometric parameters associated to $\cS(\cH)$, given in the next
table. The statement about the inradius was proved in
\cite{BarnumGurvits2005}, and the statements about the mean width
and the volume radius were obtained in \cite{AS1}.

\begin{table}[htbp] 
\caption{Radii of $\cS$ for $(\C^d)^{\otimes k}$, where $n=d^k$. }
\label{table-radii-multi}
\begin{tabular}{|r|c|c|c|c|}
  \hline
   & inradius & volume radius & mean width & outradius\\
  \hline
  $\vphantom{\displaystyle \sum} \cS((\C^d)^{\otimes k})$ & $\simeq_k n^{-1}$ &
  $\simeq_k n^{-1+\frac{1}{2k}}$ & $\simeq_k n^{-1+\frac{1}{2k}}$ & $= \sqrt{(n-1)/n}$\\
  \hline
\end{tabular}
\end{table}

The following  is a multipartite version of Theorem
\ref{threshold:separability}.

\begin{theorem}
\label{threshold:multi} For every integer $k \geq 2$, there are
effectively computable positive constants $c_k$ and $C_k$,
depending only on $k$, and a function $s_0(k,d)$ satisfying
\begin{equation} \label{eq:bounds-s0multi}
c_kd^{2k-1} \leq s_0(k,d) \leq C_kd^{2k-1} \log^2 d
,\end{equation} such that if $\rho$ is a random state on
$(\C^d)^{\otimes k}$ distributed according to the measure
$\mu_{d^k,s}$, then for every $\e>0$,
\begin{itemize}
 \item [(i)] If $s \leq (1-\e) s_0(k,d)$, then
 \[ \P(\rho \textnormal{ is separable}) \leq 2 \exp(-c(\e)s) . \]
 \item [(ii)] If $s \geq (1+\e) s_0(k,d)$, then
 \[ \P(\rho \textnormal{ is entangled}) \leq 2 \exp(-c(\e)s) . \]
\end{itemize}
Here $c(\e)$ is a positive constant depending on $\e$ and on $k$.
\end{theorem}

The proof of Theorem \ref{threshold:multi} is completely parallel
to the bipartite case, except for one  point where a
 slightly finer analysis is required.
 We set $\SEP=\SEP((\C^d)^{\otimes k})=\cS\big((\C^d)^{\otimes k}\big)-\Id/n$,
 $\K=-\SEP \cap \SEP$ and  define the
threshold $s_0(k,d)$ as
\[ s_0(k,d) := \left(\frac{\E \|G\|_{\SEP}}{d^k} \right)^2 \sim w(\SEP^\circ)^2. \]

We first need to show that the estimates \eqref{eq:bounds-s0multi}
hold. We only sketch the proof. First, Proposition \ref{prop:symmetrization} carries
over verbatim to the multipartite setting, and implies (using the formula from Table
\ref{table-radii-multi}) that
\[ \E \|G\|_{\K^\circ} \simeq_k n^{\frac{1}{2k}}. \]
The multipartite version of Lemma \ref{lem:product:mean:width} is

\begin{lemma} \label{lem:multi}
In the notation of the present section, we have
\[ n^2 \lesssim \E \|G\|_{\K} \cdot \E \|G\|_{\K^\circ} \lesssim_k n^2 \log n.\]
\end{lemma}

\begin{proof}
We proceed as follows. Let $E \subset \Hzero$ be the subspace
spanned by the products of trace zero operators, and $F$ be the
orthogonal complement of $E$. We have $\dim E=(d^2-1)^k$ and $\dim
F = d^{2k}-1-(d^2-1)^k \simeq_k n^{2-2/k}$. Let $u:\Hzero \to
\Hzero$ be a linear map such that $u(\K)$ is in the
$\ell$-position. By the results in Appendices
\ref{app:ell-position} and \ref{app:irreducible} we may assume
that $u$ has the form
\[ u = P_{E} + ( {\bf 0}_E \oplus v ) \]
for some positive map $v:F \to F$. As in the bipartite case, the
$MM^*$-estimate and the ideal property of the $\ell$-norm imply
that $\ell_{\K}(P_E)\ell_{\K^\circ}(P_E) \lesssim n^2 \log n$, and
so Lemma \ref{lem:multi} will follow from the following

\begin{claim}   \label{claim:multi} 
$\ell_{\K}(P_{F})=
o_k(\ell_{\K}(\Id))=o_k(n^{2-\frac{1}{2k}}) $ and
$\ell_{\K^\circ}(P_{F})=
o_k(\ell_{\K^\circ}(\Id))$ $=o_k(n^{\frac{1}{2k}})$.
\end{claim}

To show the Claim, we note first that the inradius of $\K$ 
equals the inradius of $\SEP$ and is of order
$n^{-1}$ (from Table \ref{table-radii-multi}).  
Hence $\|\cdot\|_{\K} \lesssim_k n|\cdot|$ and so 
\[ \ell_{\K}(P_{F}) = \E \|P_{F}G\|_{\K} \lesssim_k
n \E |P_{F} G| = n\gamma_{\dim F} \simeq_k n^{2-1/k}. \] 
This implies the first part of the Claim. The second part requires a
little finer analysis. For $1 \leq j \leq k$, denote by $F_j$ the
following subspace of $\Hzero$
\[ F_j = \Big( \vphantom{\sum_i} \cM_{sa}(\C^d) \otimes \cdots \otimes \cM_{sa}(\C^d)
\otimes \R\Id \otimes \cM_{sa}(\C^d) \otimes \cdots \otimes
\cM_{sa}(\C^d) \Big) \cap \Hzero, \] where the factor $\R \Id$
appears at position $j$. Since $F \subset \bigoplus_j F_j$, it
suffices to prove that
$\ell_{\K^\circ}(P_{F_j})=o_k(n^{\frac{1}{2k}})$---it follows from
the ideal property that $\ell_{\K^\circ}(P_{V}) \leq
\ell_{\K^\circ}(P_{V'})$ whenever $V \subset V'$. By symmetry we
may assume $j=1$. We have
\begin{equation} \label{eq:ell-Kpolar-multi} \ell_{\K^\circ}(P_{F_1})
= \E \|P_{F_1}G\|_{\K^\circ} = w_G(P_{F_1}\K) \leq w_G(P_{F_1}
\SEP ).\end{equation} It turns out that the convex body
$P_{F_1}\SEP$ has a simple description.
\begin{lemma} \label{lem:section-of-SEP}
For every $k \geq 2$, we have
\[ P_{F_1}\SEP \left((\C^d)^{\otimes k}\right)
= \SEP \left((\C^d)^{\otimes k}\right)\cap F_1 = \Id/d \otimes
\SEP \left((\C^d)^{\otimes (k-1)}\right) ,\] with the notation $x
\otimes K = \{ x \otimes y \st y \in K \}.$
\end{lemma}
\begin{proof}
The inclusions $\supset$ are immediate. Conversely, let $\rho \in
\cS\left((\C^d)^{\otimes k}\right)$. Starting from a separable
decomposition
\[ \rho = \sum \lambda_i \rho_i^{(1)} \otimes \cdots \otimes \rho_i^{(k)}, \]
we obtain
\[ P_{F_1} \rho
= \sum \lambda_i \Id/d \otimes \rho_i^{(2)} \otimes \cdots \otimes
\rho_i^{(k)} - \Id/d^{k} = \Id/d \otimes \left( \sigma -
\Id/d^{k-1} \right),\] with $\sigma \in \cS \left((\C^d)^{\otimes
(k-1)}\right)$ defined as
\[ \sigma = \sum \lambda_i \rho_i^{(2)} \otimes \cdots \otimes \rho_i^{(k)}. \]
It follows that $P_{F_1}\SEP \left((\C^d)^{\otimes k}\right) =
P_{F_1}\cS \left((\C^d)^{\otimes k}\right) \subset \Id/d \otimes
\SEP \left((\C^d)^{\otimes (k-1)}\right)$.
\end{proof}

\begin{remark}
In the case $k=2$, the convex body $\SEP \left((\C^d)^{\otimes
(k-1)}\right)$ should be interpreted as $\cD_0(\C^d)$.
\end{remark}

We now return to the proof of Claim \ref{lem:multi}. 
From Lemma \ref{lem:section-of-SEP}, it is very easy to compute
$w_G(P_{F_1}\SEP)$. Since $\|\Id/d\|_{2}=1/\sqrt{d}$, the convex
bodies $P_{F_1}\SEP \left((\C^d)^{\otimes k} \right)$ and
$\frac{1}{\sqrt{d}} \SEP \left((\C^d)^{\otimes (k-1)}\right)$ are
congruent, hence have the same Gaussian mean width. Since $w_G
\left(\SEP \left((\C^d)^{\otimes (k-1)}\right)\right) \simeq_k
\sqrt{d}$ (see Table \ref{table-radii-multi}), we obtain from
\eqref{eq:ell-Kpolar-multi} that
\[ \ell_{\K^\circ}(P_{F_1}) \leq w_G\left(P_{F_1}\SEP((\C^d)^{\otimes k})\right)
= \frac{1}{\sqrt{d}} w_G \left(\SEP \left((\C^d)^{\otimes
(k-1)}\right)\right) \simeq_k 1, \] 
which completes the proof of the Claim, and hence that of 
 Lemma \ref{lem:multi}. \end{proof}
The
rest of the proof of Theorem \ref{threshold:multi} follows by mimicking the
argument in the bipartite case  given in Sections 
\ref{section:threshold} and \ref{section:estimate-s0}.

\section{Generic entanglement: an alternative approach}

\label{section:alternative-argument}

 In this section we sketch an alternative proof that entanglement
 on $\C^d \otimes \C^d$ is generic
 when $s$ is sufficiently smaller than  $d^3$ (Theorem \ref{threshold:separability}(i)).
The proof also yields a similar statement addressing two different
regimes of the the multipartite case, including the one dealt with
in Theorem \ref{threshold:multi}(i). In its present form,  the
argument does not show
 that separability is generic for larger $s$. However, we present it here
 since it leads to sharp estimates in the
 ``small ball'' regime (i.e., when the probability of separability is very small),
and since it is much more straightforward.

The proof is based on analyzing directly the density
\eqref{eq:formula-density} of  $\mu_{n,s}$ and specifically on the
following

 \bl \label{Lemma:comparison}
 Let  $s \geq n$ with $\log s \ll n$.
 Then there
 exists a universal constant $C>0$, such that for any measurable
 subset $\cK\subset \cD=\cD(\C^n)$,

 \begin{eqnarray*}
{\mu_{n,s}(\cK)}^{\frac{1}{m}} \leq C
 \sqrt{\frac{s}{n}} \;
{\mu_{n,n}(\cK)}^{\frac{1}{m}},
 \end{eqnarray*}
 where $m=n^2-1 = \dim \cD$.
 \el
Assuming the lemma and remembering that $\mu_{n,n}$ is just the
Lebesgue measure (the Hilbert-Schmidt volume), it is easy to
deduce (an estimate stronger than) the assertion (i) of Theorem
\ref{threshold:separability} from known results. Indeed, if
$n=d^2$ and $\cK = \cS=\cS(\C^d \otimes \C^d)$, we can read from
Table
 \ref{table-radii} that
 $$
\mu_{n,n}(\cS)^{\frac{1}{m}}  =
  \left(\frac{{\rm vol}\, \cS }{{\rm vol}\, \cD}\right)^{\frac{1}{m}}  =
  \frac {\vrad(\cS)} {\vrad (\cD)} \simeq \frac{n^{-3/4}}{n^{-1/2}} = n^{-1/4} .
 $$
Consequently,
 $$
 \mu_{n,s}(\cS) \leq \left(C_1 \sqrt{\frac{s}{n}} \; n^{-1/4}\right)^m =
 \left(C_1^2 \frac{s}{d^3} \right)^{m/2},
 $$
and so the probability of separability is very small if
$s<C_1^{-2} n^{3/2}=C_1^{-2} d^3$. 
(Note that the exponent is here of order $d^4$, as opposed to $d^3$ 
in Theorem \ref{threshold:separability}(i).)

It is also straightforward to obtain in the same way estimates on
the threshold value of $s$ for multipartite entanglement (i.e.,
for all $k\geq 2$). Sharp bounds for $\frac{{\rm vol}\, \cS }{{\rm
vol}\, \cD}$ in the multipartite case were derived in
\cite{Sz1,AS1}: \bprop \label{sepratio} Let $\cH =
\big(\C^d\big)^{ \otimes k}$ and let $ \cS(\cH)$ be the
corresponding set of $k$-partite separable states. Then
\begin{itemize}\item [(i)] $\left(\frac{{\rm vol}\, \cS(\cH) }{{\rm vol}\, \cD(\cH)}\right)^{\frac{1}{m}}
\lesssim (k \log k)^{1/2} \; n^{-\frac 12 +\frac1{2k}}$; \\
\item [(ii)] $\left(\frac{{\rm vol}\, \cS(\cH) }{{\rm vol}\,
\cD(\cH)}\right)^{\frac{1}{m}} \lesssim \left( \frac{d \hskip0.3mm
k \log k}{n^{1+ \beta_d}}\right)^{1/2}$, where $\beta_d = \log _d
(1+\frac{1}{d})-\frac{1}{d^2}\log _d(d+1)$.\end{itemize} \eprop
The two bounds in the Proposition reflect emphasis on two regimes,
fixed $k$ and large $d$ (small number of large subsystems) and
fixed $d$ and large $k$ (large number of small subsystems). When
combined with Lemma \ref{Lemma:comparison}, they lead immediately
to the following
\begin{theorem}
\label{threshold:alternative} In the notation of Proposition
\ref{sepratio}, let $s\geq n =d^k$ and let $\rho$ be a random
state on $\cH$ distributed according to the measure $\mu_{n,s}$.
Then $\P(\rho \in  \cS(\cH))$ is (exponentially in $m=n^2-1$)
small if
 $s \leq c(k)  n^{2-1/k}$ or   $s \leq c(k) c'(d) n^{2+\beta_d}$, where
 $c(k) \gtrsim (k \log k)^{-1}$  and $c'(d) \gtrsim d^{-1}$.
  \end{theorem}

 As the above bounds on the threshold values of $s$ may appear
 not-very-transparent, we will make them explicit in some special cases.
 First, in the bipartite case $k=2$  (two large subsystems), we again recover the
 bound of order $n^{3/2}=d^3$.  In the tripartite case $k=3$, entanglement
 is generic if  $s\leq c n^{5/3}$.   If $d=2$ and $k$ is large ($k$ qubits),
 the bound on the threshold value of $s$ is (modulo factors of smaller
 order) $n^{2+\beta_2}$, where $\beta_2 \approx  0.18872$
(we also have the {\em equality} $n^{2+\beta_2} = 2^k 3^{3k/4}$,
but the latter expression is again  not-so-transparent).  The
numerical constants implicit in the $\gtrsim$, $\lesssim$ notation
in the Proposition and in the Theorem are effectively computable
and can be recovered from the discussion following Theorems 1 and
2 in \cite{AS1} and from the proof of Lemma
\ref{Lemma:comparison}.

\begin{proof}[Proof of Lemma \ref{Lemma:comparison} (sketch)]  \
The argument follows the lines of \cite{Deping2010}, where related
questions were considered. Since, by the arithmetic-geometric mean
inequality, $\det \rho \leq n^{-n}$ for $\rho \in \cD(\C^n)$,
 formula \eqref{eq:formula-density} implies
$$
\mu_{n,s}(\cK)\! = \!\frac 1{Z_{n,s}}\! \int_\cK (\det \rho)^{s-n} 
d\rho \!\leq\! \frac 1{Z_{n,s}}\! \int_\cK n^{-n(s-n)} d\rho \!=
n^{-n(s-n)} \frac {Z_{n,n}} {Z_{n,s}} \mu_{n,n}(\cK) .
$$
Hence the Lemma reduces to showing that $\left(n^{-n(s-n)} \frac
{Z_{n,n}} {Z_{n,s}}\right)^{\frac{1}{m}} \lesssim \sqrt{\frac s n
}$ or, equivalently (we may replace $m=n^2-1$ in the exponent by
$n^2$ as long as $\log s \ll n$), \be \label{ratio} \left( \frac
{Z_{n,n}} {Z_{n,s}}\right)^{\frac{1}{n^2}}\lesssim \frac{s^{1/2}
n^{s/n}}{n^{3/2}}. \ee Explicit formulae for $Z_{n,s}$ are known
(see
 \cite{MML1, Zycz2, Zycz1, Ben1})
 \begin{eqnarray}Z_{n,s} =
 \frac{\sqrt{n}\, (2\pi)^{n(n-1)/2}}{\Gamma (sn)} \prod_{k=s-n+1}^s \Gamma(k) . 
 \label{avolume:D}
 \end{eqnarray}
 We point out that these quantities are sometimes referred to in
 the literature as ``the $\alpha$-volume'' (with $\alpha = s-n+1$).  Also,
 the normalization factors are often calculated for densities on the Weyl chamber
 $\big\{(\lambda_1,
  \cdots,\lambda_n)\in \R^n \, : \ \sum _{i=1}^n \lambda
  _i=1,\, \lambda_1\geq   \cdots\geq \lambda_n\geq 0\big\}$
  in the simplex of eigenvalues rather than for densities on the set of states;
 the two quantities differ by the factor $(2\pi)^{n(n-1)/2}/\prod _{j=1}^n \Gamma(j)$,
  equal to the measure of the corresponding flag manifold.

   What remains is a tedious but routine calculation based on the Stirling formula,
used here in the form $\ln \Gamma( x) = x \ln x -x + O(\ln x)$.
 We have
\begin{eqnarray*}
 \frac {\ln Z_{n,s}}{n^2}-\ln\sqrt{2\pi} & \sim
 & \frac 1{n^2} \left(\sum_{k=s-n+1}^s \ln \Gamma(k) - \ln \Gamma (sn)\right)\\
 & \sim & \frac 1{n^2} \left(\sum_{k=s-n+1}^s (k  \ln k -k) - sn \ln(sn)+sn \right)\\
  & \sim & \frac 1{n^2}\int_{k=s-n}^s( x  \ln x - x) \, dx  - \frac sn \ln(sn) +  \frac sn\\
  &=& \frac{\frac{s^2}2 \ln s\! -\! \frac{3s^2}4\! -\!
  \Big(\!\frac{(s-n)^2}2 \ln (s\!-\!n)\! -\! \frac{3(s-n)^2}4 \!\Big)}{n^2}\! -\!
  \frac sn \ln(sn)\!+\! \frac{s}n.
\end{eqnarray*}
We now use the bound $\ln (s-n) = \ln s + \ln (1-\frac ns) \leq
\ln s -\frac ns$ to obtain, after simplifications,
$$
 \frac {\ln Z_{n,s}}{n^2} \gtrsim \ln\sqrt{2\pi}  -\frac 14 -\frac sn \ln n -\frac 12 \ln s
$$
and hence, after exponentiating,
$$
({Z_{n,s}})^{\frac{1}{n^2}} \gtrsim n^{-s/n}s^{-1/2} .
$$
Similarly, but in a much simpler way, we are led to
$$
({Z_{n,n}})^{\frac{1}{n^2}}  \sim {\sqrt{2\pi}} \, e^{1/4}
n^{-3/2}.
$$
(This calculation was already performed in \cite{Sz1}, where the
equivalent formula $\vrad \cD(\C^n) \sim e^{-1/4} n^{-1/2}$ was
derived.) Combining the last two estimates we obtain
\eqref{ratio}.
\end{proof}

\section{Miscellaneous remarks and loose ends}  \label{misc}

\subsection{Threshold from zero to nonzero probability of separability} 
\label{sec:other-threshold}

In this paper we estimated the threshold for separability in terms
of the ancilla dimension, and showed that the probability of
entanglement changes dramatically from nearly one to nearly zero  
around this threshold. A seemingly related question (but actually
very different) is to ask for which ancilla dimensions the
probability of separability is {\em exactly} zero. Here is a summary of
our knowledge about this problem. (We do not claim originality.)

\begin{proposition} \label{never}
Let $\rho$ be a random state on $\C^d \otimes \C^d$ distributed
according to the probability measure $\mu_{d^2,s}$.
\begin{enumerate}
 \item[(i)] If $s \geq d^2$, then $0 < \P(\rho \textnormal{ is separable}) < 1$.
 \item[(ii)]  If $s \leq (d-1)^2$, then $\P(\rho \textnormal{ is separable}) =0$.
  \item[(iii)] If $d=2$,  then $0<\P(\rho \textnormal{ is separable})<1$ for $s\geq 3$
  and  $\P(\rho \textnormal{ is separable}) =0$ for $s\leq 2$.
\end{enumerate}
\end{proposition}
\begin{proof}
(i) follows from both $\cS$ and $\cD \setminus \cS$ having
nonempty interior and from the density of $\mu_{n,s}$
\eqref{eq:formula-density} being  strictly positive in the
interior of $\cD$ for $s\geq n$. (ii) is a simple combination of
Corollary 3.5 in \cite{Walgate}, which asserts that a random
subspace of dimension $s$ in $\C^d \otimes \C^d$  almost surely
contains no product vector when $s \leq (d-1)^2$, and of
\cite{Ho2}, which points out that a separable state must have a
product vector in its range. Next, the only instances of (iii)
which are not covered by the two preceding parts are $s=2, 3$;
these are slightly more delicate. If $s=3$, then the relevant
measure $\mu_{4,3}$ (note that $n=d^2=4$ here) is concentrated on
the boundary of $\cD$. However, since the eigenvalues of
$MM^\dagger$ (for a $4\times 3$ matrix $M$) are the same as those
of $M^\dagger M$ plus an additional $0$, and since the
distributions of $MM^\dagger$ and $M^\dagger M$ are unitarily
invariant, it follows that $MM^\dagger$ has a density with respect
to the surface measure which is (modulo a normalization factor)
$\det\big(M^\dagger M\big)$, and in particular is nonzero on a set
of full (surface) measure. On the other hand, it follows from
\cite{SBZ} that the probability of separability on the boundary of
$\cD$ (with respect to the surface measure) equals $\frac 12
\frac{\vol \cS}{\vol \cD} \in (0,\frac 12)$, which combined with
the preceding remark yields the conclusion. Similarly, if $d=2$
and $s=2$, the relevant measure $\mu_{4,2}$ is supported on the
set of states of rank (at most) $2$. The question of generic
separability of such states was studied in \cite{RW} (see also
\cite{Arv}).  While the measure considered in \cite{RW} is
apparently different from $\mu_{4,2}$, they are both induced by
parametrizations of the set of states of rank $2$ which are smooth
outside of a subset of lower dimension, and all such measures are
mutually absolutely continuous. Accordingly, our conclusion
follows from Corollary 4 in \cite{RW}.   (The authors thank Mary
Beth Ruskai for bringing the paper \cite{RW} to their attention.)
\end{proof}

\subsection{{The unbalanced case: $\cH = \C^{d_1} \otimes \C^{d_2}$, $d_1\neq d_2$}}
\label{unbalanced} 
A result analogous to Theorem
\ref{threshold:separability} holds, with the threshold
$s_0(d_1,d_2)$ verifying
$$
cn\min\{d_1,d_2\} \leq s_0(d_1,d_2) \leq Cn (\log n)^2
\min\{d_1,d_2\},
$$
where $n=d_1d_2$. To show this, let us try to retrace the
arguments from the balanced case. We may assume $d_1 \leq d_2$.
First, one checks that  the argument from  \cite{AS1} yields
$w(\cS)=w(\SEP) \simeq d_2^{1/2}/n = (n d_1)^{-1/2}$ (note that
if we are only interested in the bipartite case, the $\lesssim$
part follows rather easily from Lemma  2 in \cite{AS1}).   The
needed estimates for $w(\SEP^\circ)$ follow then word by word,
except that the upper estimate requires a slightly more careful
analysis which we detail now. If, as in the balanced case, we
denote $\K=-\SEP \cap \SEP$, we have
\[ w_G(\K) \simeq \sqrt{d_2}. \]
Let $E$, $F_1$ and  $F_2$ be the subspaces appearing in the proof
of Lemma \ref{lem:product:mean:width}. The conclusion will follow
if we prove that (for $i=1,2$)
 \begin{equation} \label{eq:unbalanced}
 \ell_{\K}(P_{F_i}) \leq \frac{1}{3} \ell_{\K}(\Id)
 = \frac{1}{3} w_G(\K^\circ)\ \ \  ; \ \ \ \ell_{\K^\circ}(P_{F_i}) \leq \frac{1}{3} \ell_{\K^\circ}(\Id)
 = \frac{1}{3} w_G(\K). \end{equation}
Indeed, since $\ell_{\K}(\Id) \leq \ell_{\K}(P_E) +
\ell_{\K}(P_{F_1}) + \ell_{\K}(P_{F_2})$, it follows from
\eqref{eq:unbalanced} that $\ell_{\K}(\Id) \leq 3 \ell_{\K}(P_E)$,
and similarly for $\ell_{\K^\circ}$, and therefore
\[ w_G(\K)w_G(\K^\circ)
= \ell_{\K}(\Id) \ell_{\K^\circ}(\Id) \leq 9 \ell_{\K}(P_E)
\ell_{\K^\circ}(P_E) \lesssim n^2 \log n \] by the
$MM^*$-estimate.

Using Proposition \ref{prop:ell-norm} and the unbalanced version
of Lemma \ref{lem:section-of-SEP}, we obtain
\[ \ell_{\K}(P_{F_1})= w_G((\K \cap F_1)^\circ)
= w_G \left(\left[ d_1^{-1/2}\cD_{\textnormal{sym}}(\C^{d_2})
\right]^\circ \right) \simeq \sqrt{d_1}\, d_2^{3/2}, \] where
$\cD_{\textnormal{sym}}=-\cD_0 \cap \cD_0$. Similarly,
\[ \ell_{\K^\circ}(P_{F_1}) = w_G(P_{F_1}\K) \leq
w_G (d_1^{-1/2}\cD(\C^{d_2})) \simeq\sqrt{d_2/d_1}. \] Note that
the equivalences $w_G(\cD_{\textnormal{sym}}^\circ) \simeq
w_G(\cD_0^\circ)$ and $w_G(\cD_{\textnormal{sym}}) \simeq
w_G(\cD)$, and the values of these parameters, were determined in
Proposition \ref{prop:symmetrization}, (i) and (iii).
Analogous estimates hold for $F_2$ with the role of
$d_1$ and $d_2$ exchanged.
 We conclude that the conditions \eqref{eq:unbalanced} are satisfied unless
 $d_1$ is smaller than some absolute constant.

On the other hand, if $2\leq \min\{d_1,d_2\} \leq C$, then
  $\cD_0 \subset (C^2-1) \SEP$
  (cf. \cite{BarnumGurvits2002}, Corollary 5; by looking at the
  Schmidt decomposition of a pure state it is easy to see
  that any bound valid with $d_1=d_2= C$ -- the case considered
  in \cite{BarnumGurvits2002} -- holds also when $ \min\{d_1,d_2\} \leq C$).
 Accordingly, the widths of $\cS, \SEP^\circ$ are the same as those of  $\cD, \cD_0^\circ$
 ``up to universal multiplicative constants''  and the calculation becomes trivial
 (in this case, there is no $\log$ factor in the upper bound for $s_0(d_1, d_2)$).

\subsection{Calculating the precise order of the mean width of $\cS^\circ$, or is the $\log$
necessary}\label{log} It is conceivable that the presence of
logarithmic factors in Lemma  \ref{lem:product:mean:width}, and
hence in Theorems \ref{threshold:separability} and
\ref{threshold:multi}, is due to the fact that we appeal to a
result about general convex bodies (Proposition \ref {thm:MMstar})
rather than calculate   $w(\SEP^\circ)$ directly (or to the lack
of precision in the upper bound from Proposition \ref
{thm:MMstar}). Due to the fundamental nature of the
separability/entanglement dichotomy, the question about the
precise order of   $w(\SEP^\circ)$ is interesting by itself. It
may be worthwhile to note that for the convex body
$L=\cD_0(\C^n)$, the product $w(L)w(L^\circ)$ is majorized by a
universal constant, even though the argument via $\ell$-position
necessarily leads to an upper bound which is  $\gtrsim \sqrt{\log
n}$.

\subsection{Thresholds for other ``standard'' properties of quantum states}\label{other}

Our method can be generalized to estimate thresholds for other
properties of random induced states (beyond separability),
provided the set of states with this property is a convex subset
$K \subset \cD$ and has some minimal invariance properties (such
as being fixed under conjugation with local unitaries). In any
such application one needs to estimate the mean width of
$K^\circ$, or at least the volume of $K$ in order to adapt the
argument from Section \ref{section:alternative-argument}.

One natural example of $K$ is the set 
${\mathcal{PPT}}(\C^d \otimes
\C^d)$ of states with positive partial transpose (the volume
radius and the mean width of which were estimated in \cite{AS1}
and shown to be much larger than those of $\cS$). Since
${\mathcal{PPT}} = \cD \cap T(\cD)$, where $T$ is the partial
transpose, it follows easily (cf. Proposition \ref{prop:symmetrization}(i)) that
$w\big({{\mathcal{PPT}}}_0^\circ \big) \leq 2 w(\cD_0^\circ) \simeq
d$, whence (cf. \eqref{eq:def-s0}) $s_0 \simeq d^2$. However, this
is less precise than the result from \cite{Aubrun2011} where it
was shown -- using completely different techniques -- that a sharp
threshold for PPT occurs at $s_0 \sim 4d^2$. (While a more careful
argument using concentration of measure shows that
$w\big({{\mathcal{PPT}}}_0^\circ \big)  \sim 2d$, in our approach we
will always lose an additional multiplicative constant when using
Proposition \ref{prop:MP-vs-GUE}. Similarly, while it is likely
that our argument based on majorization may be modified to show
{\em existence} of a sharp threshold also in the regime $s_0 \simeq
d^2$, it is not at all clear how to retrieve this way the exact
multiplicative factor $4$.)

A very interesting point is that there is a whole range of
parameters, when the ancilla dimension $s$ satisfies (for an
arbitrary $\e >0$ and an appropriate $c>0$)
\[ (4+\e) d^2 \leq s \leq cd^3, \]
where random induced states on $\C^d \otimes \C^d$ are -- for
large $d$ -- generically {\itshape bound entangled}, i.e.
entangled and PPT (hence non-distillable  \cite{hhh}).

\subsection{\mbox{Improving the probability estimates, 
and sharpness  of the threshold.}}
\label{sharpness}

Denote (as in Section \ref{section:threshold}) by $\pi_{d,s}$ the
probability that a random state on $\C^d \otimes \C^d$ with
distribution $\mu_{d^2,s}$ is separable.  In Section
\ref{section:threshold} we showed that, if $s \leq $ $(1-\e) s_0(d)$,
then \eqref{sbound} holds, i.e.,
$$ \pi_{d,s} \leq 2  \exp(-c(\e)s).$$

The above probability estimate is not optimal, both in its
dependence on $s$ and in the dependence on $\e$ that can be
retrieved from the argument. The following conjecture sounds
reasonable (the larger the environment is, the more uncommon
entanglement is) and would formally imply the bound  $\pi_{d,s}
\leq 2\exp(-c(\e)s_0(d))$, which better agrees with physical
heuristics and is formally stronger than both \eqref{sbound} and
the assertion of Theorem \ref{threshold:separability} (i).

\begin{conjecture}
For any $d \geq 2$, the function $s \mapsto \pi_{d,s}$ is
non-increasing.
\end{conjecture}

While we do not prove this conjecture, the next simple lemma is a
partial result in this direction.

\begin{lemma} \label{lem:monotonicity2}
For every $d,s$, we have the inequality $\pi_{2d,s} \leq
\pi_{d,4s}$.
\end{lemma}

\begin{proof}
Identify $\C^{2d}$ with $\C^2 \otimes \C^d$, and let $\tau :
\cD(\C^{2d} \otimes \C^{2d}) \to \cD(\C^{d} \otimes \C^d)$ be the
partial trace over $\C^2 \otimes \C^2$. Let $\rho$ be a random
state on $\C^{2d} \otimes \C^{2d}$ with distribution
$\mu_{4d^2,s}$. Then $\tau(\rho)$ is a random state on $\C^d
\otimes \C^d$ with distribution $\mu_{d^2,4s}$, and
\[ \rho \text{ separable } \Longrightarrow \tau(\rho) \text{ separable.} \]
This shows the inequality $\pi_{2d,s} \leq \pi_{d,4s}$.
\end{proof}

There are other ways to improve the probability estimate in
various ranges. First, note that $\pi_{d,s} =0$ if $s \leq
(d-1)^2$ (see Section \ref{sec:other-threshold}). Second, Theorem
\ref{threshold:alternative} implies that, for some absolute
constants $c, c_1>0$, $\pi_{d,s} \lesssim \exp(-cd^4)$ whenever
$d^2 \leq s \leq c_1d^3$. This establishes in particular the bound
$\pi_{d,s}  \leq 2 \exp(-c(\e)d^3 )$ asserted in Theorem
\ref{threshold:separability} (i), except for the narrow range
$(d-1)^2 <s< d^2$. This exceptional range can be handled as
follows. Set $d_1=\lfloor (d+1)/2\rfloor$; then $\pi_{d,s} \leq
\pi_{d_1,s}$ by Lemma \ref{lem:monotonicity}. On the other hand,
$s> (d-1)^2 \geq d_1^2$  and (if $d$ is sufficiently large) $s <
d^2 \leq c_1d_1^3$, so Theorem \ref{threshold:alternative} applies
and shows that $\pi_{d_1,s} \lesssim \exp(-cd_1^4) \leq
\exp(-cd^4/16)$. Combining the estimates we conclude that the
bound $\pi_{d,s} \lesssim \exp(-cd^4)$ extends to the entire range
$s \leq c_1d^3$ (perhaps with a different constant $c$). When
combined with the argument from Section \ref{section:threshold},
this completes the proof of Theorem \ref{threshold:separability}
(i) as originally stated.

Estimates for probabilities in Theorem
\ref{threshold:separability} and similar statements translate
directly into assertions about sharpness of the
entanglement-separability threshold at $s_0=s_0(d)$:  the increase
in the ancilla dimension $s$ needed for the induced state to
``flip''  from generic entanglement to generic separability. As
stated,  Theorem \ref{threshold:separability}  asserts that the
increase is $o(s_0)$. Retracing the argument would allow to come
up with an explicit (and clearly suboptimal) bound, apparently of
the order of $s_0^\theta$ for some $\theta \in (0,1)$ (with
$\theta$ rather close to $1$).  On the other hand, finding precise
order is likely a difficult ``small ball'' problem (cf.
\cite{LiShao}). So, instead of pursuing such calculations, we
sketch a heuristic argument which  suggests the limits of our
approach,  which may be not far from the actual behavior.

Our analysis shows that the sharpness of the threshold $s_0$
depends on a combination of two effects: (i) the decay of
$\E\|\rho_{d^2,s}- \frac{\Id}{n}  \|_{\SEP}$ as a functions of $s$
and (ii) the concentration of $\|\rho_{d^2,s}- \frac{\Id}{n}
\|_{\SEP}$ around its mean (or median). We do know  (from
\eqref{MPorder}) that the former is approximately $\phi(s):=
\sqrt{\frac{s_0}{s}}$, and from Proposition
\ref{prop:concentration}  (see \eqref{concentration}) that
$\|\rho_{d^2,s}- \frac{\Id}{n}  \|_{\SEP}$ is concentrated on an
interval whose length is of order $\frac1{\sqrt{s}}$: choosing as
$\eta:=\frac{a}{\sqrt{s}}$, where $a$ is sufficiently large, makes
the right hand side of \eqref{concentration} small.  A simple
calculation shows now that the increase in  $s$ needed to reduce
the value of $\phi(s)$ from $1+\eta$ to $1-\eta$ is of order
$s^{1/2}$.  The lack of rigor in this calculation stems from the
approximation given by $\phi(s)$  not being known  to be valid up
to the precision of order $\eta$.

\subsection{Regularity of the threshold function, and proof of Corollary \ref{particles}.}
\label{proof:particles} Theorem \ref{threshold:separability} is a
statement about sharp transition
 from generic entanglement to generic separability as
the dimension of the ancilla $s$ increases while the dimension
$n=d^2$ of the system itself remains fixed. Clearly, this implies
that a similar phenomenon takes place also  with fixed $s$ and
variable $d$. However, without any additional information about
regularity of
 the threshold function $s_0(d)$,  one can not formally infer that -- for example -- 
 a sharp transition occurs also in this new setting.

While we do not have a complete picture of the regularity of
$s_0(d)$, or of the associated probabilities $\pi_{d,s}$,  we do
have {\em some} information
 (Lemmas \ref{lem:monotonicity} and \ref{lem:monotonicity2}), which is
 enough to deduce  Corollary \ref{particles}.

As noted already in the Introduction,  the setting  of Corollary
\ref{particles} (i.e.,   $N$ particles with $D$ levels each and
two subsystems of $k$ particles each)
 is modeled by a random induced  state on $\C^d \otimes \C^d$
with $d=D^k$  and $s = D^{N-2k}$.  Thus we need to show that the
sequence $p_k := \pi_{D^k,D^{N-2k}}$ has the following property

(i) {\sl for some small $\delta >0$ (which quantifies the ``near
$0$'' and ``near $1$'' probabilities) and for some $k_0\sim N/5$
(the threshold value of $k$), $p_k > 1- \delta$ if $k<k_0$  and
$p_k  < \delta$ if $k>k_0$. } 

\noindent Except for determining the value of $k_0$, this
is equivalent to the following 

(ii)  {\sl if  $p_k \leq 1-\delta$,
then $p_{k+1} < \delta$.} 

\noindent We note that a slight generalization of
Lemma \ref{lem:monotonicity2}, with $2$ replaced by an arbitrary
$D$ and $4$ by $D^2$,  implies that the sequence $(p_k)$ is
non-increasing; this is not necessary for our argument, but
nevertheless reassuring. Also, if $k$ is substantially  smaller
than $N/5$, then $N-2k$ is substantially larger than $3k$ and so
$s = D^{N-2k}$ is substantially larger than $D^{3k}=d^3$ and,
consequently -- by Theorem \ref{threshold:separability} -- $p_k$
is close to $1$. Similarly, if $k$ is substantially  larger than
$N/5$,  $p_k$ is close to $0$. Accordingly, there is no doubt that
the transition from $p_k \approx 1$ to $p_k \approx 0$  does
indeed occur as $k$ increases, and that it occurs when $k\sim
N/5$. The only point that needs to be made is the sharpness of the
transition.

To that end, note that -- by Theorem \ref{threshold:separability}
--- a property similar to (i) and (ii) holds for the sequence
$q_i:= \pi_{d,D^i}$ for any fixed $d$: if $D^i \leq (1-\e)
s_0(d)$, then  $\pi_{d,D^i} < \delta$ and if $D^i \geq (1+\e)
s_0(d)$, then $\pi_{d,D^i} >1- \delta$ (as long as $\delta >
2\exp(-c(\e)s)$ for the appropriate values of $s$), and $ (1-\e)
s_0(d) < D^i <  (1+\e) s_0(d)$ may happen at most for one value of
$i  \sim \log_D(s_0(d))\sim 3\log_D d$ (there will be no such $i$
at all, unless $\log_D(s_0(d))$  is close to an integer).

To show that the condition (ii) above is satisfied for the
sequence  $(p_k)$, we argue as follows. If $p_k =
\pi_{D^k,D^{N-2k}} \leq 1-\delta$, then -- by Lemma
\ref{lem:monotonicity2} -- also $\pi_{D^{k+1},D^{N-2k}} \leq
1-\delta$, and so the observation from the preceding paragraph
applied with $q_i=\pi_{D^{k+1},D^{i}}$ implies that
$\pi_{D^{k+1},D^{N-2k-1}} < \delta$ and, applied one more time
(this works  if $\delta \leq 1-\delta$, or if $\delta \leq 1/2$, which 
may be readily assumed),
 that $p_{k+1} = \pi_{D^{k+1},D^{N-2k-2}} < \delta$, as needed.

\appendix

\section{Majorization and $\iy$-Wasserstein distance}
\label{app:wasserstein}

We gather here some facts concerning the usual {\em modes of
convergence} from probability, {\em $\infty$-Wasserstein
distance}, and  the concept of {\em majorization}. They will be
used in the proof of Proposition  \ref{prop:MP-vs-GUE} in the next
appendix, but they are fairly general and independent of that
particular application; we believe that stating them separately
may be of reference value.

\begin{definition}
Let $\mu_1,\mu_2$ be probability measures on $\R$. The
$\infty$-Wasserstein distance is defined as
$$d_{\infty}(\mu_1,\mu_2):=\inf \|X_1-X_2\|_{L^\infty},$$
with infimum over all couples $(X_1,X_2)$ of random variables with
(marginal) laws $\mu_1$ and  $\mu_2$,  defined on a common
probability space. Similarly, if $Y_1, Y_2$ are real random variables,
we will mean  by $d_{\infty}(Y_1, Y_2)$  the $\infty$-Wasserstein
distance between the laws of $Y_1$ and $ Y_2$.
\end{definition}

Note the following inequality: whenever $f : \R \to \R$ is a
$L$-Lipschitz function and $X,Y$ are bounded random variables, we
have
\begin{equation} \label{eq:lipschitz-wasserstein} 
|\E f(X)-\E f(Y)| \leq L \hskip.5mm d_{\iy}(X,Y).
\end{equation}
The $\iy$-Wasserstein distance can be computed from cumulative
distribution functions: if $F_X(t)=\P(X \leq t)$, then
\[ d_{\iy}(X,Y) = \inf \{ \e>0 \st F_X(t-\e) \leq F_Y(t) \leq F_X(t+\e) \textnormal{ for all } t \in \R \} . \]

The following lemma is elementary and can be proved by using the
fact that the L\'evy distance metrizes the weak convergence for
probability measures (see e.g. \cite{galambos}, Section 4.3).

\begin{lemma} \label{lem:wasserstein}
Let $Z$ be a random variable distributed according to a measure
$\mu_Z$, with support {\bf equal} to some bounded interval
$[a,b]$. If  $(Y_n)$ is a sequence of random variables,  the
following are equivalent:
\begin{enumerate}
\item $ d_{\infty}(Y_n, Z) \to 0$,
\item $Y_n \to Z$ weakly and 
$\sup Y_n  \to  b$,  $\inf Y_n  \to  a$. \footnote{By $\inf$ and $\sup$
we really mean here {\em essential} $\inf$ and $\sup$.}
\end{enumerate}
\end{lemma}

Note that the hypothesis on the support is crucial: the
equivalence fails if the support is not connected.

Next, we will relate the $\infty$-Wasserstein distance to the
concept of {\em majorization}, usually defined for  vectors in
$\R^n$. Given  $x\in \R^n$ we will  denote by $x^\downarrow$ the
non-increasing rearrangement of $x$.   If $x,y \in \R^{n,0}$ (the
hyperplane of sum $0$ vectors in $\R^n$), we write $x \prec y$ if,
for every $k \in \{1,\dots,n\}$, we have
\[ \sum_{i=1}^k x^\downarrow_i \leq \sum_{i=1}^k y^\downarrow_i \]
(note that for $k=n$ we always have equality).  The following is
well-known (see \cite{bhatia}, Section II).

\begin{proposition} \label{prop:majorization}
For  $x,y \in \R^{n,0}$ the following are equivalent.
\begin{enumerate}
\item $x \prec y$,
\item whenever $\phi$ is a permutationally invariant
convex function on $\R^{n,0}$, then  $\phi(x) \leq \phi(y)$,
\item For every $t \in \R$, we have $ \sum_{i=1}^n |x_i-t| \leq \sum_{i=1}^n |y_i-t|$,
\item $x$ can be written as a convex combination of coordinatewise  permutations of $y$.
\end{enumerate}
\end{proposition}

We will need a {\em quantitative} version of the concept of
majorization. If $x,x' \in \R^{n,0}$ with $x' \neq 0$, we denote
by $\delta(x,x')$ the smallest non-negative constant $t$ such that $x
\prec t  x'$. In other words, $\delta(\cdot,x')$ is the gauge
associated to the convex body obtained as the convex hull of
coordinatewise permutations of $x'$. The quantity $\delta(x,x')$ should
not be thought of as a distance, but rather as the
 norm of a certain operator. We have the inequality $\delta(x,x'') \leq
\delta(x,x')\delta(x',x'')$. More generally, if $\phi$ is 
any $1$-homogeneous convex function on $\R^{n,0}$ 
(for example, the gauge of a convex body) 
which is permutationally-invariant, then 
\be \label{majorization}
\phi(x) \leq \delta(x,x') \phi(x').
\ee
We can rephrase the concept of majorization in the language of
 ``$n$-point empirical measures,'' i.e.,
the probability measures of the form   $\nu_x = n^{-1}
\sum_{k=1}^n \delta_{x_k}$. The restriction requiring that $x\in
\R^{n,0}$ translates into the underlying random variable having
zero mean. Note that (by Proposition \ref{prop:majorization})
\[\delta(x,y)\!=\delta(y,x) \!= 1\!\! \iff \!\!y 
\textnormal{ is a coordinatewise  permutation of } x \!\!\iff\!\! \nu_x\!=\!\nu_y. \]
The following key observation connects majorization and
convergence in $\iy$-Wasserstein distance, and may be of independent interest.

\begin{proposition}  \label{lemma:convergence-to-Z}
Let $Z$ be a non-constant bounded random variable with mean $0$
and with distribution $\mu_Z$. Then, for every $\e>0$, there
exists $\eta>0$ (depending on $Z$ and $\e$) such that for every
$n$ and for all vectors $x,y \in \R^{n,0}$ satisfying
$d_{\iy}(\nu_x,\mu_Z) \leq \eta$ and $d_{\iy}(\nu_y,\mu_Z) \leq
\eta$, we have $\delta(x,y) \leq 1+\e$.
\end{proposition}

\begin{proof}
Denote $a=\inf Z$ and $b =\sup Z$ (the hypotheses on $Z$ imply
$a<0<b$). Let $\e >0$, and $x,y \in \R^{n,0}$ such that
$d_{\iy}(\nu_x,\mu_Z) \leq \eta$ and $d_{\iy}(\nu_y,\mu_Z) \leq
\eta$. We must choose $\eta$ to ensure that $\delta(x,y) \leq
1+\e$ or, equivalently, $x \prec (1+\e)y$.  By Proposition
\ref{prop:majorization}, this is equivalent to
$$
\sum_{i=1}^n |x_i-t| \leq \sum_{i=1}^n |(1+\e)y_i-t|  \quad \hbox{
for all } t \in \R
$$
(where $x = \big(x_i\big)_{i=1}^n$ and $y =
\big(y_i\big)_{i=1}^n$)  or, in the language of $n$-point
empirical measures,
\begin{equation}
f_0(t) := \int |u-t| \, d\nu_{x}(u) \leq \int |(1+\e)u-t| \,
d\nu_{y}(u) =: g_0(t).  \label{deffngn}
\end{equation}
We first note that these conditions are automatically satisfied if
$t\geq \max_i x_i$. Indeed, we have then  $|u-t| = t-u$ for  all
$u$ in the support of $\nu_x$ and so the first integral above
equals $\int t d\nu_x - \int u d\nu_x = t-0=t$. On the other hand,
we always have $ |(1+\e)u-t|  \geq  t-(1+\e)u$ and so the second
integral is at least $\int t d\nu_y - \int (1+\e)u d\nu_y =
t-(1+\e)0=t$. Similar argument applies when $t\leq \min_i x_i$.

If we choose $\eta \leq \frac{\e}{2}\min(-a,b)$, then we have
$(1+\e/2)a \leq \min_i x_i \leq  \max_i x_i \leq (1+\e/2)b$ and
therefore it suffices to show the inequality $f_0 \leq g_0$ on the
interval $\big((1+\e/2)a,(1+\e/2)b\big)$. To that end, we compare
the functions $f_0$ and $g_0$ with the functions
\[ f(t) = \E | Z-t | = \int |u-t| \, d\mu_{Z}(u) ,\]
\[ g(t) = \E | (1+\e)Z-t | = \int |(1+\e)u-t| \, d\mu_{Z}(u) .\]
Proposition \ref{lemma:convergence-to-Z} will now follow from the
following lemma.

\begin{lemma} \label{lem:fg}
In the above notation, $f(t) < g(t)$ for $t \in
\big((1+\e)a,(1+\e)b\big)$.
\end{lemma}

Indeed, since $f$ and $g$ are continuous, there is a number
$\theta >0$ such that $f(t)<g(t)-\theta$ for every $t \in
\big[(1+\e/2)a,(1+\e/2)b\big]$. On the other hand, by
\eqref{eq:lipschitz-wasserstein}, we have $|f_0(t)-f(t)| \leq
\eta$ and $|g_0(t)-g(t)| \leq (1+\e)\eta$ for any $t \in \R$. It
remains to choose $\eta > 0$ such that $(2+\e)\eta \leq \theta$ to
guarantee that $f_0 \leq g_0$ on $\big[(1+\e/2)a,(1+\e/2)b\big]$.
\end{proof}

\begin{proof}[Proof of Lemma \ref{lem:fg}]

Assume first $t\geq0$.  We have
$$
g(t)-\!f(t) \!=\!\!\! \int\!\!\! \big( |(1+\e)u-t|-|u-t| \big) d\mu_{Z}(u)\! =\!\!\!
\int\!\!\! \big( |(1+\e)u-t|-|u-t| +\e u\big) d\mu_{Z}(u) .
$$
It is now elementary to check that the last integrand is $0$ if
$u\leq  t/(1+\e)$ and strictly positive if $u >  t/(1+\e)$.
Accordingly, the integral is always nonnegative, and it is
strictly positive if the interval $ (t/(1+\e), \infty)$ intersects
the support of $\mu_Z$, that is exactly when $t/(1+\e) < b$, or
$t<(1+\e)b$.  The case  $t\leq 0$ is handled similarly (or by a
change of variable $u=-v$).
\end{proof}

\section{GUE approximation to induced states: Proposition
 \ref{prop:MP-vs-GUE}} \label{app:marcenko-pastur-vs-GUE}

The strategy of the proof is as follows. First, we gather known
facts from Random Matrix Theory which assert that, for the regime
in question (i.e., $n, \frac ns \to \infty$), the appropriately
normalized ensembles $\big(G_n\big)$ and $\big(\rho_{n,s} -
\frac{\Id}{n}\big)$ converge (in probability) to the same limit in
the sense of non-commutative probability. Then we will show that
-- in the same regime -- the expectations of the gauges
$\|\cdot\|_K$ must be asymptotically the same for both ensembles,
which is essentially the assertion of Proposition
\ref{prop:MP-vs-GUE}; this part of the argument will be based on
the material from Appendix \ref{app:wasserstein}.

We first set some notation. Recall that $\cM^{sa}_n$ is the space
of self-adjoint operators on $\C^n$ and $\cM^{sa,0}_n$ -- the
subspace of self-adjoint trace $0$ operators. 
We denote by $\spec(A) \in \R^n$ the spectrum of an operator $A
\in \cM^{sa}_n$ (ordered in the increasing order for definiteness,
but this is irrelevant). Note that $A \in \cM^{sa,0}_n \iff
\spec(A) \in \R^{n,0}$ (the hyperplane of sum zero vectors in $\R^n$). 
Conversely, if $x \in \R^{n}$, we denote
by $\diag(x) \in \cM_n^{sa}$ the diagonal matrix built from $x$.

Recall that the standard {\em semicircular distribution} (or Wigner
distribution) is the probability distribution $\mu_{sc}$ with
density
\[ \frac{1}{2\pi} \sqrt{4-x^2} {\bf 1}_{|x| \leq 2} .\]
Note that a random variable $Z$ with semicircular distribution
satisfies the hypotheses of Lemmas  \ref{lem:wasserstein} and
Proposition \ref{lemma:convergence-to-Z}.

We are now in a position to state the needed facts from Random
Matrix Theory. Recall that the $n$-point empirical measure
associated to a vector $x=(x_k) \in \R^n$ is the measure $\nu_x=
n^{-1} \sum_{k=1}^n \delta_{x_k}$.

\begin{proposition} \label{pr:RMT1}
For every $n$, let $G_n$ be an $n \times n$ $\textnormal{GUE}^0$
random matrix, and let $\nu_{\spec (n^{-1/2}G_n)}$ be   the
rescaled empirical spectral distribution. When $n$ tends to
infinity, the sequence $\big(\nu_{\spec (n^{-1/2}G_n)}\big)$
converges to $\mu_{sc}$, in probability, with respect to the
$\iy$-Wasserstein distance.
\end{proposition}

 Convergence in probability means that for any $\e>0$, $\lim_{n \to
\iy} \P( d_{\iy}(\mu_n,\mu_{sc}) $ $> \e) = 0.$

\begin{proposition} \label{pr:RMT2}
For every $n,s$, let $A_{n,s}=\rho_{n,s} - \frac{\Id}{n}$, where
$\rho_{n,s}$ is a random state with distribution $\mu_{n,s}$, and
let $\nu_{\spec ( \sqrt{ns} A_{n,s)}}$ be the rescaled empirical
spectral distribution. When $n$ and $s/n$ tend to infinity, the
sequence $\big(\nu_{\spec (\sqrt{ns} A_{n,s)}}\big)$ converges to
$\mu_{sc}$, in probability, with respect to the $\iy$-Wasserstein
distance.
\end{proposition}

\begin{proof}[References for Propositions \ref{pr:RMT1} and \ref{pr:RMT2}]
By Lemma \ref{lem:wasserstein}, convergence to the semi-circle
distribution with respect to the $\iy$-Wasserstein distance is
equivalent to weak convergence and convergence of extreme
eigenvalues. These statements appear separately in random matrix
literature.

Wigner's famous result (the semi-circle law, see \cite{agz}) is that 
$\nu_{\spec(n^{-1/2}G_n)}$ converges weakly, in probability, towards
$\mu_{sc}$. The standard setting is the case of GUE matrices,
while here we consider $\textnormal{GUE}^0$ matrices (i.e., GUE
matrices conditioned to have trace $0$), but it is easily checked
that this doesn't affect the limit distribution.

The statement about convergence of limit eigenvalues is also
well-known. For example, it is known that $\E \|G_n\|_{\iy} \leq
2\sqrt{n}$ (a proof for $\textnormal{GUE}$ matrices can be found
in Appendix H of \cite{Sz1}, and this extends to
$\textnormal{GUE}^0$ via Jensen's inequality). Concentration of
measure (e.g., the Gaussian version of Lemma  \ref{lemma:Levy})
implies then that
\begin{equation} \label{eq:tailestimate} 
\P( \|G_n/\sqrt{n}\|_{\iy} \geq 2+t) \leq \exp(-nt^2/2) .
\end{equation}

Some readers may be surprised by  Proposition \ref{pr:RMT2} since 
$\rho_{n,s}$ is a (rescaled) Wishart ensemble, which is 
known to have  the  Marchenko--Pastur law as the limiting spectral distribution
\cite{Marchenko-Pastur 1967}. 
However, that limit is obtained when $n,s \to \infty$ with the ratio 
$s/n \to \beta$ for some $\beta >0$, while here we have $s/n \to \infty$. 
For a non-rigorous but convincing and simple calculation, 
the reader may verify that when $\beta \to \infty$, the (appropriately rescaled) 
Marchenko--Pastur density does indeed converge to a 
(non-centered) semi-circular density. 

For a rigorous argument, we refer to  \cite{BaiYin}, where
the weak convergence towards the semicircular distribution is
proved (one should point out that this statement appeared
implicitly already in \cite{Marchenko-Pastur 1967}). The statement
about limit eigenvalues, as well as tail inequalities analogous to
\eqref{eq:tailestimate}, can be deduced from the techniques from
\cite{HT} (see also \cite{CNY}, Theorem 2.7). 

We point out that while the (rather difficult) memoir \cite{HT} is the 
ultimate reference on the subject of rectangular complex Gaussian matrices,   
and while it is indispensable for recovering very sharp results, 
the estimates we need here can be obtained by much simpler methods. 
The primary reason we are invoking \cite{HT} is that, historically, 
the topic of  rectangular random matrices was studied only in the real case 
because of its relevance to statistics. However, most -- 
but not all (cf. \cite{ds-handbook}, \S 2.3) -- 
arguments  carry over to the complex case; 
in particular,  the simple and elegant proof  from \cite{silverstein} 
would be worthwhile to analyze in this regard. 
Finally, let us note that once the behavior of the limit eigenvalues 
is determined, the tail inequalities of the type \eqref{eq:tailestimate} 
follow from the Gaussian concentration (see the comments following 
Lemma \ref{lemma:Levy}). 
\end{proof}

For every integer $n$, denote by $X_n^{\text{sc}} \in \R^{n,0}$
the ``ideally semicircular'' vector, i.e. the vector
$X_n^{\text{sc}}=(X_{n,1}^{\text{sc}},\dots,X_{n,n}^{\text{sc}})$
such that
\[ F_{\text{sc}}(X_{n,k}^{\text{sc}})= \frac{2k-1}{2n}, \]
where $F_{\text{sc}}$ is the cumulative distribution function of
standard semicircular distribution (note that the sum of
coordinates is indeed zero). Obviously (see the beginning of 
Appendix \ref{app:wasserstein}), the sequence
$\big(\nu_{X_n^{\text{sc}}}\big)$ converges to $\mu_{sc}$ in the
$\iy$-Wasserstein distance.

Let $K$ be a convex body in $\cM^{sa,0}_n$ (containing $0$ in the
interior), with $\|\cdot\|_K$ the corresponding gauge function.
Define a gauge $\phi_K$ on $\R^{n,0}$ by setting
\[ \phi_K(x) = \int_{\cU(n)} \| U \diag(x) U^\dag \|_K dU , \]
where the integral is taken with respect to the (normalized) Haar
measure on the unitary group. Let $G_n$ be an $n \times n$
$\text{GUE}^0$ random matrix, and let $U \in \cU(n)$ be a
Haar-distributed random unitary matrix independent from $G_n$. By
unitary invariance, $G_n$ has the same distribution as $U^\dagger
\diag(\spec(G_n)) U$. Therefore, we have
\[ \E \phi_K(\spec(G_n)) = \E \|G_n\|_K .\]

Since $\phi_K$ is convex and invariant under permutation of
coordinates, it follows from \eqref{majorization} that  
$\phi_K(x) \leq \delta(x,y) \phi_K(y)$ for every
$x,y \in \R^{n,0}$ . In particular, if we introduce the random
variables $\alpha_n=\delta
(\spec(\frac{1}{\sqrt{n}}G_n),X_n^{\text{sc}} )$ and
$\beta_n=\delta (X_n^{\text{sc}},\spec(\frac{1}{\sqrt{n}}G_n))$,
we have
\[ \sqrt{n} \phi_K(X_n^{\text{sc}}) \beta_n^{-1} \leq \phi_K(\spec(G_n))
 \leq \sqrt{n} \phi_K(X_n^{\text{sc}}) \alpha_n .\]

Taking expectation, we obtain
\begin{equation} \label{ineq-alphabeta1}
\sqrt{n} \phi_K(X_n^{\text{sc}}) \E \beta_n^{-1} \leq \E \|G_n\|_K
\leq \sqrt{n} \phi_K(X_n^{\text{sc}}) \E \alpha_n .
\end{equation}
Recall that $A_{n,s} \!=\! \rho_{n,s}-\frac{\Id}{n}$, where $\rho_{n,s}$ is a
random state with distribution $\mu_{n,s}$, and introduce the random
variables $\alpha'_{n,s}\!\!=\!\delta
(\spec(\!\sqrt{ns}A_{n,s}),X_n^{\text{sc}} )$ and
$\beta'_{n,s}\!\!=\!\delta (X_n^{\text{sc}},\spec(\!\sqrt{ns}A_{n,s})\!)$.
Since the distribution of $A_{n,s}$ is unitarily invariant, the
same argument applies verbatim, and yields the inequalities
\begin{equation} \label{ineq-alphabeta2}
\frac{1}{\sqrt{ns}} \phi_K(X_n^{\text{sc}}) \E (\beta'_{n,s})^{-1}
\leq \E \|\rho_{n,s}-\frac{\Id}{n}\|_K \leq \frac{1}{\sqrt{ns}}
\phi_K(X_n^{\text{sc}}) \E \alpha'_{n,s} .\end{equation}

By combining \eqref{ineq-alphabeta1} and \eqref{ineq-alphabeta2},
one obtains the following inequalities for the constants $C_{n,s}$
and $c_{n,s}$ introduced in Proposition \ref{prop:MP-vs-GUE}
\[ \frac{\E (\beta'_{n,s})^{-1}}{\E \alpha_n} \leq  c_{n,s} \leq C_{n,s} 
\leq \frac{\E \alpha'_{n,s}}{\E \beta_n^{-1}}\]

It is now immediate to deduce Proposition \ref{prop:MP-vs-GUE}
from the following lemma.

\begin{lemma}  \label{lem:majorizationRMT}
In the notation introduced above, we have
\[ \lim_{n \to \iy} \E \alpha_n = \lim_{n \to \iy} \E \beta_n^{-1} = \lim_{n,s/n \to \iy} \E \alpha'_{n,s} =
 \lim_{n,s/n \to \iy} \E (\beta'_{n,s})^{-1} = 1 .\]
\end{lemma}

\begin{proof}[Proof of Lemma \ref{lem:majorizationRMT}]
Propositions \ref{pr:RMT1} and \ref{pr:RMT2} combined with Proposition 
\ref{lemma:convergence-to-Z} imply that
$\alpha_n$, $\beta_n$, $\alpha'_{n,s}$ and $\beta'_{n,s}$ converge in
probability to $1$ when $n$ and $s/n$ tend to infinity.

The convergence of $\E \alpha_n$ will follow if we show that, say,
$\sup_n \E \alpha_n^2 < +\iy$. Introduce a vector $z \in \R^{n,0}$
with $\lfloor n/2 \rfloor$ coordinates equal to $1$ and $\lfloor
n/2 \rfloor$ coordinates equal to $-1$ (if $n$ is odd, the
remaining coordinate is necessarily $0$). It is easily checked
that $\delta(z,X_n^{\text{sc}})$ is bounded by an absolute
constant $C$. Moreover, for any $x \in \R^{n,0}$, we have
$\delta(x,z)=\|x\|_{\iy}$. It follows that
\[ \alpha_n \leq C \frac{1}{\sqrt{n}} \| G_n \|_{\iy} \]
and the problem is reduced to showing that $\sup_n \E
\|G_n/\sqrt{n}\|_{\iy}^2 < +\iy$, which follows easily from the
tail estimates \eqref{eq:tailestimate}.

Since $\alpha_n \beta_n \geq 1$,  convergence of $\E \beta_n^{-1}$
follows from the convergence of $\E \alpha_n$. The convergence of
$\E \alpha'_{n,s}$ and $\E (\beta'_{n,s})^{-1}$ is proved along
the same lines.
\end{proof}
\begin{remark} The above arguments are quite flexible and 
become particularly simple if we do not
aim at showing that $\delta(\cdot,\cdot) \sim 1$, but at an
estimate $\delta(\cdot,\cdot) \lesssim 1$. For example, the
estimates stated in the second remark following Proposition
\ref{prop:MP-vs-GUE} can be deduced from the known bounds on $\E
\| \cdot \|_{\iy}$ (such as \eqref{eq:tailestimate}) for the
ensembles in question and the following elementary fact:
 if $x, y \in \R^{n,0}$ are such that $\max_j |x_j| \leq \frac Cn \sum_j|y_j|$, then $x  \prec 2Cy$.
\end{remark}

\section{Irreducible subspaces for the isometry group of $\cS$} \label{app:irreducible}

Fix an integer $d>1$. Denote by $H$ the (real) vector space
$\cM_d^{sa}$, $H_0=\Hzero$ the hyperplane of trace zero matrices,
and $H_1=H_0^\perp$ the one-dimensional space of scalar matrices.
(Note that this notation is not identical to the one used in the
main body of the paper.)

The unitary group $\cU(\C^d)$ acts on $H$ by conjugacy: the action
of an element $U \in \cU(\C^d)$ is given by
\[  A \mapsto UAU^\dagger . \]

Similarly, for any integer $k$, the direct product $\mU(\C^d)^k$
acts on $H^{\otimes k}$: the action of a $k$-tuple
$(U_1,\dots,U_k)$ being given by
\begin{equation} \label{eq:action}
A \mapsto (U_1 \otimes \cdots \otimes U_k) A(U_1^\dagger \otimes
\cdots \otimes U_k^\dagger)
\end{equation}
(this construction is called {\em external tensor product}  in
representation theory).

 When a group $\Gamma$ acts on a (real or complex) vector space $E$ (i.e., if $\mL(E)$
 is the space of linear operators on $E$ and
 $\pi : \Gamma \to \mL(E)$ is  a representation), one
 says that a nonzero subspace $F \subset E$ is {\itshape invariant} 
 if, for every $\gamma \in \Gamma,
 \pi(\gamma) F \subset F$. 
 We say that a nonzero subspace $F \subset E$ is {\itshape irreducible} if,
 for every nonzero $x \in  F$,
 \[ F=\Span \{ \pi(\gamma) x \st \gamma \in \Gamma \}; \]
 in other words, if $F$ is invariant, but no proper subspace of $F$ is.

\begin{lemma} \label{lem:irreducible}
Consider the action of $\mU(\C^d)^k$ on $H^{\otimes k}$ given by
\eqref{eq:action}. A subspace $E \subset H^{\otimes k}$ is
irreducible if and only if it has the form
\begin{equation} \label{eq:elementary-subspace}
H_{i_1} \otimes \cdots \otimes H_{i_k} \end{equation} for some
choice of $(i_1,\dots,i_k) \in \{0,1\}^k$. Moreover, a subspace $E
\subset H^{\otimes k}$ is invariant if and only if it is the
direct sum of subspaces of the form
\eqref{eq:elementary-subspace}.
\end{lemma}

Before we prove Lemma \ref{lem:irreducible}, let us state basic
results from representation theory (see \cite{Serre}).

\begin{proposition} \label{thm-representations}
The following results are valid for {\bfseries complex}
representations of compact groups.
\begin{enumerate}
\item[(i)] For every representation $\pi : \Gamma \to \mL(E)$, there exists a decomposition
of $E$ as direct sum
\begin{equation} \label{eq:decomposition} E = \bigoplus_{\alpha} E_\alpha ,\end{equation}
where $(E_\alpha)$ are irreducible subspaces.
\item[(ii)] If, moreover, there do not exist indices $\alpha \neq \alpha'$ such that
the subrepresentations of $\pi$ into $\mL(E_\alpha)$ and
$\mL(E_{\alpha'})$ are isomorphic, then the decomposition in
\eqref{eq:decomposition} is unique.
\item[(iii)] If $\pi_1 : \Gamma_1 \to \mL(E_1)$ and $\pi_2 : \Gamma_2 \to \mL(E_2)$
are two irreducible representations, then the (external) tensor
product representation
\[ \pi_1 \otimes \pi_2 : \Gamma_1 \times \Gamma_2 \to \mL(E_1 \otimes E_2) \]
is irreducible.
\end{enumerate}
\end{proposition}

The statements (i), (ii) and (iii) appear in \cite{Serre} as
Theorems 2, 8 and 10 respectively (note that although these
theorems are stated in \cite{Serre} for finite groups, they remain
valid for compact groups, as noted in Chapter 4).

\begin{proof}[Proof of Lemma \ref{lem:irreducible}]
Let us switch to the complex field: denote by $E^\C$ the
complexification of a real vector space $E$. Note that $H^\C$
naturally identifies with the space of complex $d \times d$
matrices, and $(H^{\otimes k})^\C$ naturally identifies with
$(H^\C)^{\otimes k}$.

We first check that the subspaces $H_0^\C$ and $H_1^\C$ are
irreducible for the action of $\mU(\C^d)$ on $H$. Obviously
$H_1^\C$ is irreducible. Let $A \in H_0^\C$ with $A \neq 0$ and
consider the (complex) space
\[ \mF = \Span \{ U^\dagger AU \st U \in \mU(\C^d) \} .\]
We must show that $\mF=H_0^\C$. Note that $\mF$ necessarily
contains a matrix with the diagonal non identically zero (since
one may diagonalize either the Hermitian or the anti-Hermitian
part of $A$, and one of them is nonzero). Averaging over all
diagonal unitary matrices shows then that $\mF$ contains a nonzero
diagonal matrix. Since the symmetric group $\mathfrak{S}_{d}$ acts
irreducibly on the hyperplane of sum zero vectors in $\C^{d}$, we
deduce that $A$ contains every diagonal trace zero matrix, and
therefore every Hermitian or skew-Hermitian trace zero matrix, so
that $\mF=H_0^\C$.

By Proposition \ref{thm-representations}(iii), it follows that for
every $k \in \N$ and every $(i_1,\dots,i_k) \in \{0,1\}^k$, the
subspace
\[ E_{(i_1,\dots,i_k)} = H^\C_{i_1} \otimes \cdots \otimes H^\C_{i_k} \]
is an irreducible subspace in $(H^{\C})^{\otimes k}$ for the
action of $\mU(\C^d)^k$. Moreover, the $2^k$ corresponding
subrepresentations are pairwise non-isomorphic (indeed, for every
$1 \leq j \leq k$, the index $i_j$ can be retrieved by looking at
the subgroup corresponding to the $j$th copy of $\mU(\C^d)$: its
action on $E_{(i_1,\dots,i_k)}$ is trivial iff $i_j=1$).
 By Proposition \ref{thm-representations}(ii), this is the unique decomposition of
$(H^\C)^{\otimes k}$ into irreducible subspaces, and any invariant
subspace is obtained as the direct sum of some of these
irreducible subspaces.

Note that a subspace $E \subset H^{\otimes k}$ is invariant if and
only if its complexification $E^\C \subset (H^\C)^{\otimes k}$ is
invariant. This implies the second part of Lemma
\ref{lem:irreducible}. The first part follows since the
irreducible subspaces are minimal in the lattice of invariant
subspaces.
\end{proof}

\section{The $\ell$-position} \label{app:ell-position}

{\bfseries Important:} In this section we restrict our attention
to {\bf symmetric} convex bodies.
 While most of the concepts and  results generalize to the non-symmetric case,
 it is not known whether the central result,  Proposition \ref{thm:MMstar}, holds in that setting.

Our presentation of $\ell$-position is standard and follows
closely the book \cite{Pisier1989}. Let $K \subset \R^m$ be a
symmetric convex body, and let $T : \R^m \to \R^m$ be a linear
operator. We define the quantity $\ell_K(T)$ as
\[ \ell_K(T) = \E \|T(G)\|_K, \]
where $G$ denotes the standard Gaussian vector in $\R^m$ (i.e. the
coordinates are independent $N(0,1)$ random variables). If there
is no ambiguity about  the underlying convex body, we write $\ell$
instead of $\ell_K$.  The following proposition collects
elementary properties of this concept.

\begin{proposition} \label{prop:ell-norm} If $K \subset \R^m$ is a symmetric convex body, then
\begin{enumerate}
 \item[(i)]  $\ell_K(\cdot)$ is a norm on $\mL(\R^m)$,
 \item[(ii)]   $\ell_K(T_1T_2) \leq \ell(T_1)_K \|T_2\|_{op}$ and
 $\ell_K(T_1T_2) \leq \ell_K(T_2) \|T_1\|_{op}$ (this is called
 the ideal property),
 \item[(iii)]
$ \ell_K(\Id) = w_G(K^\circ) = \gamma_m w(K^\circ) \sim \sqrt{m}\,
w(K^\circ)$ and  $\ell_K(u) = \ell_{u^{-1}K}(\Id)$,
 \item[(iv)] If $P_E$ denotes the orthogonal projection on a subspace $E \subset \R^m$, then
 \[ \ell_K(P_E)=w_G((K \cap E)^\circ)=w_G(P_EK^\circ),  \]
 where by $(K \cap E)^\circ$ we mean the polar inside $E$.
\end{enumerate}
\end{proposition}
\noindent Before proceeding,  let us point out that the more
common definition of the $\ell$-norm is via the second moment,
namely $(\E \|T(G)\|_K^2)^{1/2}$. However, the two expressions are
equivalent.

\begin{proposition} \label{prop:gaussian-moments}
For any symmetric convex body $K \subset \R^m$ and for any linear
operator $T : \R^m\to \R^m$, we have
\[ \E \|T(G)\|_K \leq (\E \|T(G)\|_K^2)^{1/2}\leq \sqrt{\frac{\pi}{2}} \E \|T(G)\|_K .\]
\end{proposition}

Proposition \ref{prop:gaussian-moments} (which we do not really
need here, but include for clarity) is a special case of Corollary
3 in \cite{LO}. If we do not insist on obtaining the optimal
constant $\sqrt{\pi/2}$, the result is more elementary
(essentially a special case of the so-called Khinchine--Kahane
inequality) and extends to non-symmetric convex bodies (see, e.g.,
\cite{flatness}, Lemma 3.3).

We now consider the maximization problem
\begin{equation} \label{maximization} \max_{\ell(T) \leq 1} \det T .\end{equation}
By compactness, the maximum is attained and it is obviously
strictly positive. Since, for $V \in O(m)$ and for any $T$, we
have $\ell(T) = \ell(TV)$, it follows that it is enough to
restrict our attention to $T$'s that are positive definite (PD).
We claim that, under the PD restriction, the solution $T_0$ 
to \eqref{maximization} is unique. 
Indeed, if we had any other solution $T_1$, it would
follow that $T=(T_0+T_1)/2$ verifies, on the one hand, $\ell(T)
\leq 1$ and, on the other hand, $\det T > (\det T_0)^{1/2}  (\det
T_1)^{1/2}= \det T_0$ (by strict log-concavity of $\det$ over
PD), a contradiction. Clearly, $\ell (T_0)~=~1$.

If the maximum above is attained when $T$ is a multiple of
identity, one says that $K$ is in the $\ell$-position. For every
symmetric convex body $K$, there is a linear transformation $T$
such that $TK$ is in the $\ell$-position; moreover $T$ is unique
up to rotations and homotheties.

We will take advantage of the uniqueness of the $\ell$-position
through the following lemma
\begin{lemma} \label{lem:unique-ell}
Let $K \subset \R^m$ be a symmetric convex body and $\Gamma$ be
the
 isometry group of $K$ (i.e. the set of orthogonal transformations $U$ such that $UK=K$).

Then there is a linear map $T$ such that $TK$ is in the
$\ell$-position and which has the form
\[ T = \sum_i \lambda_i P_{E_i}, \]
for some $\lambda_i>0$ and some subspaces $(E_i)$ which are
invariant under the action of $\Gamma$.
\end{lemma}

\begin{proof}
Let $T \geq 0$ be the unique solution to the maximization problem
\eqref{maximization}. For any $\gamma \in \Gamma$ and $x \in
\R^m$, we have $\| \gamma(Tx) \|_K = \|Tx\|_K$,  hence
$\ell(\gamma T)=\ell(T)$. Since $\ell(TV)=\ell(T)$ for any
orthogonal transformation $V$, it follows that $\ell(\gamma
T\gamma^{-1})=\ell(T)$. Uniqueness of the solution implies that
$\gamma T\gamma^{-1}=T$. Write
\[ T= \sum_i \lambda_i P_i, \]
where $\lambda_i>0$ are distinct positive numbers and $P_i$
pairwise orthogonal projectors. Given $i$, we have $\gamma
P_i\gamma^{-1}=P_i$ for all $\gamma \in \Gamma$, in particular
 the range of $P_i$ is invariant under the action of $\Gamma$.
\end{proof}

We will use the following theorem which compares the mean width of
$K$ and the mean width of $K^\circ$ whenever $K$ is a convex body
in the $\ell$-position.

\begin{proposition} \label{thm:MMstar}
Let $K \subset \R^m$ be a symmetric convex body in the
$\ell$-position. Then
\[ \ell_K(\Id)\, \ell_{K^\circ}(\Id) \lesssim m \log m ,\]
or, equivalently,
\[ w(K) \,w(K^\circ) \lesssim \log m .\]
\end{proposition}

This deep result is known in the Asymptotic Geometric Analysis as the
``$MM^*$-estimate". It follows by combining results of Figiel and
Tomczak-Jaegermann with those of Pisier, including in particular
sharp bounds on the so-called ``$K$-convexity constant.'' See
chapter 3 in \cite{Pisier1989} for a complete proof and
references.

Incidentally, while $O(\log m)$ is the optimal general upper bound
for the $K$-convexity constant mentioned above (see
\cite{Bourgain1984}), to the best of our knowledge it is not known
 whether it gives the correct order in Proposition  \ref{thm:MMstar}.
The pair $(\ell_1^m,\ell_{\infty}^m)$ gives an example for which
$w(K)w(K^\circ)$ is of order $\sqrt{\log m}$. In the non-symmetric
case, the $m$-dimensional simplex $\Delta$ is an example with
$w(\Delta)w(\Delta^\circ)$ $\simeq \log m$, but known general upper
bounds for non-symmetric convex bodies are much weaker
\cite{flatness, Rudelson}.

Note that the lower bound $ w_G(K) w_G(K^\circ) \geq \gamma_m^2 \sim m$ 
(or, equivalently,  $
w(K) w(K^\circ) \geq 1$) follows simply from
\begin{eqnarray*}
\gamma_m &=&  \E |G| \leq \E \big(\|G\|^{1/2}_K\|G\|^{1/2}_{K^\circ}\big) \leq
\left( \E \|G\|_K \right)^{1/2}\left( \E \|G\|_{K^\circ}
\right)^{1/2}\\&=&\gamma_m w(K^\circ)^{1/2}w(K)^{1/2} .\end{eqnarray*}

{\small 
\vskip6mm \noindent {\bf Acknowledgements.} Part of this research was performed
during the Fall of 2010 while SJS and DY
visited the Fields Institute in Toronto, Canada
(``Thematic Program on Asymptotic Geometric Analysis"), and while
GA and SJS visited Institut Mittag-Leffler
(``Quantum Information Theory'' programme).
Thanks are due to these institutions and to their staff  for their hospitality, and to 
fellow participants for many inspiring interactions. 
SJS thanks Harsh Mathur for providing insight into physicist's view 
of quantum theory. 

The research of GA was supported in part by the
{\itshape Agence Nationale de la Recherche} grants ANR-08-BLAN-0311-03 
and ANR 2011-BS01-008-02. 

The research of SJS was supported in part by grants from the
{\itshape National Science Foundation (U.S.A.)}, from the {\itshape U.S.-Israel
Binational Science Foundation}, and by the second ANR grant listed under GA. 

The research of DY was 
supported by the Fields Institute, the {\itshape NSERC} Discovery Accelerator
Supplement Grant \#315830 (from Carleton University), and by  {\itshape ERA} 
and  {\itshape NSERC} discovery grants (from the University of Ottawa), while
DY held a Fields-Ontario postdoctoral fellowship.
}


\bibliographystyle{plain}

\small

\vskip 10mm \noindent  Guillaume Aubrun, {\small \em Institut
Camille Jordan, Universit\'{e} Claude Bernard Lyon 1, 43 Boulevard
du 11 Novembre 1918, 69622 Villeurbanne CEDEX, France.}\\
{\small \tt E-mail: aubrun@math.univ-lyon1.fr}

\vskip 2mm \noindent  Stanis\l aw  Szarek,
{\small \em Case Western Reserve University,
Cleveland, Ohio 44106-7058, USA} {\small and} \\  {\small \em
Institut de Math\'ematiques de Jussieu, Universit\'{e}  Pierre et Marie Curie, 75005 Paris, France }\\
{\small\tt Email: szarek@cwru.edu}

\vskip 2mm \noindent Deping Ye, {\small \em School of Mathematics
and Statistics, Carleton University, Ottawa, ON, K1S5B6, Canada,}
 {\small and {\em The Fields Institute, Toronto, ON, M5T3J1, Canada.}} 
 
\noindent {\small Current address:  {\em Department of
Mathematics and Statistics, Memorial University
of Newfoundland, St. John's, Newfoundland, Canada A1C 5S7}} \\
 {\small \tt Email: deping.ye@mun.ca}

\end{document}